%% file: IT_Journal.tex
\documentclass[journal,12pt,draftclsnofoot, onecolumn]{IEEEtran}
\usepackage{ifpdf}
\usepackage{cite}
\usepackage[pdftex]{graphicx}
\usepackage{amsmath}
\usepackage{algorithmic}
\usepackage{array}
\usepackage{subfigure}
\usepackage{caption}
\usepackage[caption=false,font=footnotesize]{subfig}
\usepackage{fixltx2e}
\usepackage{stfloats}
\usepackage{url}
\usepackage{amsthm}
\usepackage{amsmath}
\usepackage{amsfonts}
\usepackage{amssymb}
\usepackage[T1]{fontenc} 
\usepackage{amsmath}
\usepackage[cmintegrals]{newtxmath}
\usepackage{bm} 
\usepackage[all]{xy}
\usepackage{enumitem}
\usepackage{float}
\usepackage{amsmath}
\usepackage[usenames, dvipsnames]{color}
\usepackage{ifpdf}
\usepackage{cite}
\usepackage[pdftex]{graphicx}
\usepackage{amsmath}
\usepackage{algorithmic}
\usepackage{array}
\usepackage{mathtools}
\DeclarePairedDelimiter{\abs}{\lvert}{\rvert}

\usepackage{caption}
\usepackage[caption=false,font=footnotesize]{subfig}
\usepackage{fixltx2e}
\usepackage{stfloats}
\usepackage{url}
\usepackage{mathtools}
\usepackage{amsthm}
\usepackage{amsmath}
\usepackage{amsfonts}
\usepackage{amssymb}
\usepackage[T1]{fontenc} 
\usepackage{amsmath}
\usepackage[cmintegrals]{newtxmath}
\usepackage{bm} 
\usepackage[all]{xy}
\usepackage{enumitem}
\usepackage{float}
\usepackage{amsmath}
\usepackage[dvipsnames]{xcolor}

\theoremstyle{definition}

\newtheorem{thm}{Theorem}

\newtheorem{example}{Example}
\newtheorem{lem}{Lemma}

\newtheorem{define}{Definition}

\newcommand{\no}{\nonumber}

\begin{document}
\title{Matching Anonymized and Obfuscated\\ Time Series to Users' Profiles}

\author{Nazanin~Takbiri~\IEEEmembership{Student Member,~IEEE,}
        Amir~Houmansadr~\IEEEmembership{Member,~IEEE,}
        Dennis~Goeckel~\IEEEmembership{Fellow,~IEEE,}
        Hossein~Pishro-Nik~\IEEEmembership{Member,~IEEE}
        \thanks{N. Takbiri is with the Department
        of Electrical and Computer Engineering, University of Massachusetts, Amherst,
        MA, 01003 USA e-mail: (ntakbiri@umass.edu).}
        \thanks{A. Houmansadr is with the College of Information and Computer Sciences, University of Massachusetts, Amherst,
        MA, 01003 USA e-mail:(amir@cs.umass.edu)}
        \thanks{H. Pishro-Nik and D. Goeckel are with the Department
        of Electrical and Computer Engineering, University of Massachusetts, Amherst,
        MA, 01003 USA e-mail:(pishro@engin.umass.edu)}
        \thanks{This work was supported by National Science Foundation under grants CCF--0844725, CCF--1421957 and CNS1739462. 
        
        This work was presented in part in IEEE International Symposium on Information Theory (ISIT 2017) \cite{nazanin_ISIT2017}.}}

\maketitle

\begin{abstract}

Many popular applications use traces of user data to offer various services to their users.  However, even if user data is anonymized and obfuscated, a user's privacy can be compromised through the use of statistical matching techniques that match a user trace to prior user behavior. In this work, we derive the theoretical bounds on the privacy of users in such a scenario. We build on our recent study in the area of location privacy,
in which we introduced formal notions of location privacy for anonymization-based location privacy-protection mechanisms.
Here we derive the fundamental limits of user privacy when both anonymization and obfuscation-based protection mechanisms are applied to users' time series of data.
We investigate the impact of such mechanisms on the trade-off between privacy protection and user utility. We first study achievability results for the case where the time-series of users are governed by an i.i.d.\ process. The converse results are proved both for the i.i.d.\ case as well as the more general Markov chain model. We demonstrate that as the number of users in the network grows, the obfuscation-anonymization plane can be divided into two regions: in the first region, all users have perfect privacy; and, in the second region, no user has privacy.
\end{abstract}



\input{introduction}
\input{framework-1}



\input{conclusion}

\appendices

\input{appendix_a}
\input{appendix_c}
\input{appendix_b}
\input{appendix_d}

\bibliographystyle{IEEEtran}
\bibliography{REF}
%
%
%
%

 \end{document}

%% file: introduction.tex
\section{Introduction}
\label{intro}

A number of emerging systems and applications work by analyzing the data submitted by their users in order to serve them; we call such systems \emph{User-Data Driven} (UDD) services. Examples of UDD services include smart cities, connected vehicles, smart homes, and connected healthcare devices, which have the promise of greatly improving users' lives. Unfortunately, the sheer volume of user data collected by these systems can compromise users' privacy~\cite{FTC2015}. Even the use of standard Privacy-Protection Mechanisms (PPMs), specifically anonymization of user identities and obfuscation of submitted data, does not guarantee users' privacy, as adversaries are able to use powerful statistical inference techniques to learn sensitive private information of the users~\cite{0Quest2016, 3ukil2014iot, 4Hosseinzadeh2014,iotCastle,matching}.

To illustrate the threat of privacy leakage, consider three popular UDD services: (1) {\em Health care:}  Wearable monitors that constantly track user health variables can be invaluable in assessing individual health trends and responding to emergencies.  However, such monitors produce long time-series of user data uniquely matched to the health characteristics of each user; (2) {\em Smart homes:} Emerging smart-home technologies such as fine-grained power measurement systems can help users and utility providers to address one of the key challenges of the twenty-first century:  energy conservation.  But the measurements of power by such devices can be mapped to users and reveal their lifestyle habits; and, (3) {\em Connected vehicles:}  The location data provided by connected vehicles promises to greatly improve everyday life by reducing congestion and traffic accidents.  However, the matching of such location traces to prior behavior not only allows for user tracking, but also reveals a user's habits.  In summary, despite their potential impact on society and their emerging popularity, these UDD services have one thing in common: their utility critically depends on their collection of user data, which puts users' privacy at significant risk.

There are two main approaches to augment privacy in UDD services: \emph{identity perturbation (anonymization)}~\cite{1corser2016evaluating,hoh2005protecting,freudiger2007mix, ma2009location, shokri2011quantifying2, Naini2016,soltani2017towards, soltani2018invisible}, and \emph{data perturbation (obfuscation)}~\cite{shokri2012protecting, gruteser2003anonymous, bordenabe2014optimal}. In anonymization techniques, privacy is obtained by concealing the mapping between users and data, and the mapping is changed periodically to thwart statistical inference attacks that try to de-anonymize the anonymized data traces by matching user data to known user profiles. Some approaches employ $k$-anonymity to keep each user's identity indistinguishable within a group of $k-1$ other users ~\cite{2zhang2016designing,11dewri2014exploiting, gedik2005location, zhong2009distributed, sweeney2002k, kalnis2007preventing,liu2013game}.  Other approaches employ users' pseudonyms within areas called mix-zones~\cite{beresford2003location, freudiger2009optimal, palanisamy2011mobimix}.
Obfuscation mechanisms aim at protecting privacy by perturbing user data, e.g., by adding noise to users' samples of data.   For instance, cloaking replaces each user's sample of data with a larger region~\cite{18shokri2014hiding,8zurbaran2015near,hoh2007preserving, wernke2014classification, chow2011spatial, um2010advanced}, while an alternative approach is to use dummy data in the set of possible data of the users~\cite{kido2005protection, shankar2009privately, chow2009faking, kido2005anonymous, lu2008pad}. In~\cite{randomizedresponse}, a mechanism of obfuscation was introduced where the answer was changed randomly with some small probability. Here we consider the fundamental limits of a similar obfuscation technique for providing privacy in the long time series of emerging applications.


The anonymization and obfuscation mechanisms improve user privacy at the cost of user utility. The anonymization mechanism works by frequently changing the pseudonym mappings of users to reduce the length of time series that can be exploited by statistical analysis.  However, this frequent change may also decrease the usability by concealing the temporal relation between a user's sample of data, which may be critical in the utility of some systems, e.g., a dining recommendation system that makes suggestions based on the dining history of its users. On the other hand, obfuscation mechanisms work by adding noise to users' collected data, e.g., location information. The added noise may degrade the utility of UDD applications.  Thus, choosing the right level of the privacy-protection mechanism is an important question, and understanding what levels of anonymization and obfuscation can provide theoretical guarantees of privacy is of interest.


In this paper, we will consider the ability of an adversary to perform statistical analyses on time series and match the series to descriptions of user behavior. In related work, Unnikrishnan~\cite{matching} provides a comprehensive analysis of the asymptotic (in the length of the time series) optimal matching of time series to source distributions. However, there are several key differences between that analysis and the work here. First, Unnikrishnan~\cite{matching} looks at the optimal matching tests, but does not consider any privacy metrics as considered in this paper, and a significant component of our study is demonstrating that mutual information converges to zero so that we can conclude there is no privacy leakage (hence, ``perfect privacy'').  Second, the setting of~\cite{matching} is different, as it does not consider: (a) obfuscation, which is one of the two major protection mechanisms; and (b) sources that are not independent and identically distributed (i.i.d.). Third, the setting of Unnikrishnan~\cite{matching} assumes a fixed distribution on sources (i.e., classical inference), whereas we assume the existence of general (but possibly unknown) prior distributions for the sources (i.e., a Bayesian setting).  Finally, we study the fundamental limits in terms of both the number of users and the number of observations, while Unnikrishnan~\cite{matching} focuses on the case where the number of users is a fixed, finite value.

Numerous researchers have put forward ideas for quantifying privacy-protection. Shokri et al.~\cite{shokri2011quantifying, shokri2011quantifying2} define the expected estimation error of the adversary as a metric to evaluate PPMs.  Ma et al.~\cite{ma2009location} use uncertainty about users' information to quantify user privacy in vehicular networks.
To defeat localization attacks and achieve privacy at the same time, Shokri et al.~\cite{shokri2012protecting} proposed a method which finds optimal PPM for an LBS given service quality constraints.
In~\cite{6li2016privacy} and~\cite{4olteanu2016quantifying}, privacy leakage of data sharing and interdependent privacy risks are quantified, respectively. A similar idea is proposed in \cite{14zhang2014privacy} where the quantification model is based on the Bayes conditional risk. 
Previously, mutual information has been used as a privacy metric in a number of settings,~\cite{kousha3,salamatian2013hide, csiszar1996almost, calmon2015fundamental, sankar2013utility, sankarISIT,sankar, yamamoto1983source, hyposankar}. However, the framework and problem formulation for our setting (Internet of Things (IoT) privacy) are quite different from those encountered in previous works. More specifically, the IoT privacy problem we consider here is based on a large set of time-series data that belongs to different users with different statistical patterns that has gone through a privacy-preserving mechanism, and the adversary is aiming at de-anonymizing and de-obfuscating the data.

The discussed studies demonstrate the growing importance of privacy. What is missing from the current literature is a solid theoretical framework for privacy that is general enough to encompass various privacy-preserving methods in the literature. Such a framework will allow us to achieve provable privacy guarantees, obtain fundamental trade-offs between privacy and performance, and provide analytical tools to optimally achieve provable privacy. We derive the fundamental limits of user privacy in UDD services in the presence of both anonymization and obfuscation protection mechanisms.  We build on our previous works on formalizing privacy in location-based services~\cite{tifs2016, ciss2017}, but we significantly expand those works here not just in application area but also user models and settings. In particular, our previous works introduced the notion of \emph{perfect privacy} for location-based services, and we derived the rate at which an anonymization mechanism should change the pseudonyms in order to achieve the defined perfect privacy. In this work, we expand the notion of perfect privacy to UDD services in general and derive the conditions for it to hold when \emph{both} anonymization and obfuscation-based protection mechanisms are employed.

In this paper, we consider two models for users' data: i.i.d.\ and Markov chains. After introducing the general framework in Section \ref{sec:framework}, we consider an i.i.d.\ model extensively in Section \ref{perfectsec} and the first half of Section \ref{converse}. We obtain achievability and converse results for the i.i.d.\ model. The i.i.d.\ model would apply directly to data that is sampled at a low rate. In addition, understanding the i.i.d.\ case can also be considered the first step toward understanding the more complicated case where there is dependency, as was done for anonymization-only Location Privacy-Preserving Mechanisms (LPPMs) in \cite{tifs2016}, and will be done in Section \ref{subsec:markov}. In particular, in Section \ref{subsec:markov}, a general Markov chain model is used to model users' data pattern to capture the dependency of the user' data pattern over time. There, we obtain converse results for privacy for this model. In Section \ref{sec:perfect-MC}, we provide some discussion about the achievability for the Markov chain case.

\subsection{Summary of the Results}

Given $n$, the total number of the users in a network, their degree of privacy depends
on two parameters: (1) The number of observations $m=m(n)$ by the adversary per user for a fixed anonymization mapping (i.e., the number of observations before the pseudonyms are changed);  and (2) the value of the noise added by the obfuscation technique (as defined in Section~\ref{sec:framework}, we quantify the obfuscation noise with a parameter $a_n$, where larger $a_n$ means a higher level of obfuscation).  Intuitively, smaller $m(n)$ and larger $a_n$ result in stronger privacy, at the expense of lower utility for the users.

Our goal is to identify values of $a_n$ and $m(n)$ that satisfy perfect privacy in the asymptote of a large number of users ($n \rightarrow \infty$).
When the users' datasets are governed by an i.i.d.\ process, we show that
the $m(n)- a_n$ plane can be divided into two areas.  In the first area, all users have perfect privacy (as defined in Section~\ref{sec:framework}), and, in the second area, users have no privacy.
Figure~\ref{fig:region} shows the limits of privacy in the entire $m(n)- a_n$ plane.
As the figure shows, in regions $1$, $2$, and $3$, users have perfect privacy, while in region $4$ users have no privacy.
\begin{figure}[h]
	\centering
	\includegraphics[width=0.7\linewidth, height=0.6 \linewidth]{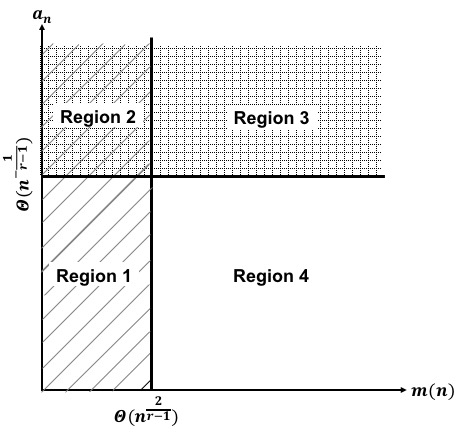}
	\caption{Limits of privacy in the entire $m(n)-a_n$ plane: in regions $1$, $2$, and $3$, users have perfect privacy, and in region $4$ users have no privacy.}
	\label{fig:region}
\end{figure}

For the case where the users' datasets are governed by irreducible and aperiodic Markov chains with $r$ states and $|E|$ edges, we show that users will have no privacy if $m =cn^{\frac{2}{|E|-r} +  \alpha}$ and $a_n =c'n^{-\left(\frac{1}{|E|-r}+\beta \right)}$, for any constants $c>0$, $c'>0$, $\alpha>0$, and $\beta>\frac{\alpha}{4}$. We also provide some insights for the opposite direction (under which conditions users have perfect privacy) for the case of Markov chains.

%% file: framework-1.tex
\section{Framework}
\label{sec:framework}
In this paper, we adopt a similar framework to that employed in~\cite{tifs2016,ciss2017}.  The general set up is provided here, and the refinement to the precise models for this paper will be presented in the following sections.  We assume a system with $n$ users with $X_u(k)$ denoting a sample of the data of user $u$ at time $k$, which we would like to protect from an interested adversary $\mathcal{A}$. We consider a strong adversary $\mathcal{A}$ that has complete statistical knowledge of the users' data patterns based on the previous observations or other resources. In order to secure data privacy of users, both obfuscation and anonymization techniques are used as shown in Figure \ref{fig:xyz}. In Figure \ref{fig:xyz}, $Z_u(k)$ shows the (reported) sample of the data of user $u$ at time $k$ after applying obfuscation, and $Y_u(k)$ shows the (reported) sample of the data of user $u$ at time $k$ after applying anonymization. The adversary observes only $Y_u(k)$, $k=1,2,\cdots, m(n)$, where $m(n)$ is the number of observations of each user before the identities are permuted. The adversary then tries to estimate $X_u(k)$ by using those observations.
\begin{figure}[h]
	\centering
	\includegraphics[width = 0.75\linewidth]{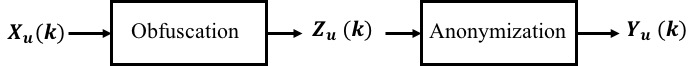}
	\caption{Applying obfuscation and anonymization techniques to users' data samples.}
	\label{fig:xyz}
\end{figure}

Let $\textbf{X}_u$ be the $m(n) \times 1$ vector containing the sample of the data of user $u$, and $\textbf{X}$ be the $m(n) \times n$ matrix with $u^{th}$ column equal to $\textbf{X}_u$;
\[\textbf{X}_u = \begin{bmatrix}
X_u(1) \\ X_u(2) \\ \vdots \\X_u(m) \end{bmatrix} , \ \ \  \textbf{X} =\left[\textbf{X}_{1}, \textbf{X}_{2}, \cdots,  \textbf{X}_{n}\right].
\]

\textit{Data Samples Model:}
We assume there are $r \geq 2$ possible values ($0,1, \cdots, r-1$) for each sample of the users' data. In the first part of the paper (perfect privacy analysis), we assume an i.i.d.\ model as motivated in Section \ref{intro}. In the second part of the paper (converse results: no privacy region), the users' datasets are governed by irreducible and aperiodic Markov chains. At any time, $X_u(k)$ is equal to a value in $\left\{0,1, \cdots, r-1 \right\}$ according to a user-specific probability distribution. The collection of user distributions, which satisfy some mild regularity conditions discussed below, is known to the adversary $\mathcal{A}$, and he/she employs such to distinguish different users based on statistical matching of those user distributions to traces of user activity of length $m(n)$.

\textit{Obfuscation Model:} The first step in obtaining privacy is to apply the obfuscation operation in order to perturb the users' data samples. In this paper, we assume that each user has only limited knowledge of the characteristics of the overall population and thus we employ a simple distributed method in which the samples of the data of each user are reported with error with a certain probability, where that probability itself is generated randomly for each user. In other words, the obfuscated data is obtained by passing the users' data through an $r$-ary symmetric channel with a random error probability. More precisely, let $\textbf{Z}_u$ be the vector which contains
the obfuscated versions of user $u$'s data samples, and $\textbf{Z}$ is the collection of $\textbf{Z}_u$ for all users,
\[\textbf{Z}_u = \begin{bmatrix}
Z_u(1) \\ Z_u(2) \\ \vdots \\Z_u(m) \end{bmatrix} , \ \ \  \textbf{Z} =\left[ \textbf{Z}_{1}, \textbf{Z}_{2}, \cdots,  \textbf{Z}_{n}\right].
\]
To create a noisy version of data samples, for each user $u$, we independently generate a random variable $R_u$ that is uniformly distributed between $0$ and $a_n$, where $a_n \in (0,1]$. The value of $R_u$ gives the probability that a user's data sample is changed to a different data sample by obfuscation, and $a_n$ is termed the ``noise level'' of the system. For the case of $r=2$ where there are two states for users' data (state $0$ and state $1$), the obfuscated data is obtained by passing users' data through a Binary Symmetric Channel (BSC) with a small error probability~\cite{randomizedresponse}. Thus, we can write
\[
{Z}_{u}(k)=\begin{cases}
{X}_{u}(k), & \textrm{with probability } 1-R_u.\\
1-{X}_{u}(k),& \textrm{with probability } R_u.
\end{cases}
\]
When $r>2$, for $l \in \{0,1,\cdots, r-1\}$:
\[
P({Z}_{u}(k)=l| X_{u}(k)=i) =\begin{cases}
1-R_u, & \textrm{for } l=i.\\
\frac{R_u}{r-1}, & \textrm{for } l \neq i.
\end{cases}
\]
Note that the effect of the obfuscation is to alter the probability distribution function of each user across the $r$ possibilities in a way that is unknown to the adversary, since it is independent of all past activity of the user, and hence the obfuscation inhibits user identification. For each user, $R_u$ is generated once and is kept constant for the collection of samples of length $m(n)$, thus, providing a very low-weight obfuscation algorithm. We will discuss the extension to the case where $R_u$ is regenerated independently over time in Section \ref{sec:perfect-MC}. There, we will also provide a discussion about obfuscation using continuous noise distributions (e.g., Gaussian noise).

\textit{Anonymization Model:} Anonymization is modeled by a random permutation $\Pi$ on the set of $n$ users. The user $u$ is assigned the pseudonym $\Pi(u)$. $\textbf{Y}$ is the anonymized version of $\textbf{Z}$; thus,
\begin{align}
\no \textbf{Y} &=\textrm{Perm}\left(\textbf{Z}_{1}, \textbf{Z}_{2}, \cdots,  \textbf{Z}_{n}; \Pi \right) \\
\nonumber &=\left[ \textbf{Z}_{\Pi^{-1}(1)}, \textbf{Z}_{\Pi^{-1}(2)}, \cdots,  \textbf{Z}_{\Pi^{-1}(n)}\right ] \\
\nonumber &=\left[ \textbf{Y}_{1}, \textbf{Y}_{2}, \cdots, \textbf{Y}_{n}\right], \ \
\end{align}
where $\textrm{Perm}( \ . \ , \Pi)$ is permutation operation with permutation function $\Pi$. As a result, $\textbf{Y}_{u} = \textbf{Z}_{\Pi^{-1}(u)}$ and $\textbf{Y}_{\Pi(u)} = \textbf{Z}_{u}$.

\textit{Adversary Model:} We protect against the strongest reasonable adversary. Through past observations or some other sources, the adversary is assumed to have complete statistical knowledge of the users' patterns; in other words, he/she knows the probability distribution for each user on the set of data samples $\{0,1,\ldots,r-1\}$. As discussed in the model for the data samples, the parameters $\textbf{p}_u$, $u=1, 2, \cdots, n$ are drawn independently from a continuous density function, $f_\textbf{P}(\textbf{p}_u)$, which has support on a subset of a defined hypercube. The density $f_\textbf{P}(\textbf{p}_u)$ might be unknown to the adversary, as all that is assumed here is that such a density exists, and it will be evident from our results that knowing or not knowing $f_\textbf{P}(\textbf{p}_u)$ does not change the results asymptotically. Specifically, from the results of Section \ref{perfectsec}, we conclude that user $u$ has perfect privacy even if the adversary knows $f_\textbf{P}(\textbf{p}_u)$. In addition, in Section \ref{converse}, it is shown that the adversary can recover the true data of user $u$ at time $k$ without using the specific density function of $f_\textbf{P}(\textbf{p}_u)$, and as result, users have no privacy even if the adversary does not know $f_\textbf{P}(\textbf{p}_u)$.

The adversary also knows the value of $a_n$ as it is a design parameter. However, the adversary does not know the realization of the random permutation $\Pi$ or the realizations of the random variables $R_u$, as these are independent of the past behavior of the users. It is critical to note that we assume the adversary does not have any auxiliary information or side information about users' data.

In \cite{tifs2016}, perfect privacy is defined as follows:
\begin{define}
User $u$ has \emph{perfect privacy} at time $k$, if and only if
\begin{align}
\no  \forall k\in \mathbb{N}, \ \ \ \lim\limits_{n\rightarrow \infty} I \left(X_u(k);{\textbf{Y}}\right) =0,
\end{align}
where $I(X;Y)$ denotes the mutual information between random variables (vectors) $X$ and $Y$.
\end{define}

\noindent In this paper, we also consider the situation in which there is no privacy. 

\begin{define}
For an algorithm for the adversary that tries to estimate the actual sample of data of user $u$ at time $k$, define
\[P_e(u,k)\triangleq P\left(\widetilde{X_u(k)} \neq X_u(k)\right),\]
where $X_u(k)$ is the actual sample of the data of user $u$ at time $k$, $\widetilde{X_u(k)}$ is the adversary's estimated sample of the data of user $u$ at time $k$, and $P_e(u,k)$ is the error probability. Now, define ${\cal E}$ as the set of all possible adversary's estimators; then, user $u$ has \emph{no privacy} at time $k$, if and only if for large enough $n$,
\[
\forall k\in \mathbb{N}, \ \ \ P^{*}_e(u,k)\triangleq \inf_{\cal E} {P\left(\widetilde{X_u(k)} \neq X_u(k)\right)} \rightarrow 0.
\]
Hence, a user has no privacy if there exists an algorithm for the adversary to estimate $X_u(k)$ with diminishing error probability as $n$ goes to infinity.
\end{define}

\textbf{\textit{Discussion:}} Both of the privacy definitions given above (perfect privacy and no privacy) are asymptotic in the number of users $(n \to \infty)$, which allows us to find clean analytical results for the fundamental limits. Moreover, in many IoT applications, such as ride sharing and dining recommendation applications, the number of users is large. 

\textbf{\textit{Notation:}} Note that the sample of data of user $u$ at time $k$ after applying obfuscation $\left(Z_u(k)\right)$ and the sample of data of user $u$ at time $k$ after applying anonymization $\left(Y_u(k)\right)$ depend on the number of users in the network $(n)$, while the actual sample of data of user $u$ at time $k$ is independent of the number of users $(n)$.  Despite the dependency in the former cases, we omit this subscript $(n)$ on $\left(Z_u^{(n)}(k), Y_u^{(n)}(k) \right)$ to avoid confusion and make the notation consistent.

\textbf{\textit{Notation:}} Throughout the paper, $X_n \xrightarrow{d} X$ denotes convergence in distribution. Also, We use $P\left(X=x\bigg{|} Y=y\right)$ for the conditional probability of $X=x$ given $Y=y$. When we write $P\left(X=x \bigg{|}Y\right)$, we are referring to a random variable that is defined as a function of $Y$.

\section{Perfect Privacy Analysis: I.I.D.\ Case}
\label{perfectsec}

\subsection{Two-States Model}

We first consider the two-states case $(r=2)$ which captures the salient aspects of the problem. For the two-states case, the sample of the data of user $u$ at any time is a Bernoulli random variable with parameter $p_u$, which is the probability of user $u$ having data sample $1$. Thus,
\[X_u(k) \sim Bernoulli \left(p_u\right).\]
Per Section \ref{sec:framework}, the parameters $p_u$, $u=1, 2, \cdots, n$ are drawn independently from a continuous density function, $f_P(p_u)$, on the $(0,1)$ interval. We assume there are $\delta_1, \delta_2>0$ such that:\footnote{The condition $\delta_1<f_P(p_u) <\delta_2$ is not actually necessary for the results and can be relaxed; however, we keep it here to avoid unnecessary technicalities.}
\begin{equation}
\no\begin{cases}
    \delta_1<f_P(p_u) <\delta_2, & p_u \in (0,1).\\
    f_P(p_u)=0, &  p _u\notin (0,1).
\end{cases}
\end{equation}

The adversary knows the values of $p_u$, $u=1, 2, \cdots, n$ and uses this knowledge to identify users. We will use capital letters (i.e., $P_u$) when we are referring to the random variable, and use lower case (i.e., $p_u$) to refer to the realization of $P_u$.

In addition, since the user data $\left(X_u(k)\right)$ are i.i.d.\ and have a Bernoulli distribution, the obfuscated data $\left(Z_u(k)\right)$ are also i.i.d.\ with  a Bernoulli distribution. Specifically,
\[Z_u(k) \sim Bernoulli\left(Q_u\right),\]
where
\begin{align}
\no {Q}_u &=P_u(1-R_u)+(1-P_u)R_u  \\
\nonumber &= P_u+\left(1-2P_u\right)R_u,\ \
\end{align}
and recall that $R_u$ is the probability that user $u$'s data sample is altered at any time. For convenience, define a vector where element $Q_u$ is the probability that an obfuscated data sample of user $u$ is equal to one, and
\[\textbf{Q} =\left[{Q}_{1},{Q}_{2}, \cdots,{Q}_{n}\right].\]
Thus, a vector containing the permutation of those probabilities after anonymization is given by:
\begin{align}
\no \textbf{V} &=\textrm{Perm}\left({Q}_{1}, {Q}_{2}, \cdots,  {Q}_{n}; \Pi \right) \\
\nonumber &=\left[ {Q}_{\Pi^{-1}(1)}, {Q}_{\Pi^{-1}(2)}, \cdots, {Q}_{\Pi^{-1}(n)}\right ] \\
\nonumber &=\left[{V}_{1}, {V}_{2}, \cdots, {V}_{n}\right ] ,\ \
\end{align}
where ${V}_{u} = {Q}_{\Pi^{-1}(u)}$ and  ${V}_{\Pi(u)} = {Q}_{u}$. As a result, for $u=1,2,..., n$, the distribution of the data symbols for the user with pseudonym $u$ is given by:
\[ Y_u(k) \sim Bernoulli \left(V_u\right) \sim Bernoulli\left(Q_{\Pi^{-1}(u)}\right)  .\]

The following theorem states that if $a_n$ is significantly larger than $\frac{1}{n}$ in this two-states model, then all users have perfect privacy independent of the value of $m(n)$.
\begin{thm}\label{two_state_thm}
 For the above two-states model, if $\textbf{Z}$ is the obfuscated version of $\textbf{X}$, and $\textbf{Y}$ is the anonymized version of $\textbf{Z}$  as defined above, and
\begin{itemize}
	\item $m=m(n)$ is arbitrary;
	\item $R_u \sim Uniform [0, a_n]$, where $a_n \triangleq c'n^{-\left(1-\beta\right)}$ for any $c'>0$ and $0<\beta<1$;
\end{itemize}
then, user 1 has perfect privacy. That is,
\begin{align}
\no  \forall k\in \mathbb{N}, \ \ \ \lim\limits_{n\rightarrow \infty} I \left(X_1(k);{\textbf{Y}}\right) =0.
\end{align}
\end{thm}

The proof of Theorem \ref{two_state_thm} will be provided for the case $0\leq p_1<\frac{1}{2}$, as the proof for the case $\frac{1}{2}\leq p_1\leq1$ is analogous and is thus omitted.
\\

\noindent \textbf{Intuition behind the Proof of Theorem \ref{two_state_thm}:} 

Since $m(n)$ is arbitrary, the adversary is able to estimate very accurately (in the limit, perfectly) the distribution from which each data sequence $\textbf{Y}_u$, $u= 1, 2, \cdots, n$ is drawn; that is, the adversary is able to accurately estimate the probability $V_u$, $u= 1, 2, \cdots, n$. Clearly, if there were no obfuscation for each user $u$, the adversary would then simply look for the $j$ such that $p_j$ is very close to $V_u$ and set $\widetilde{X_j(k)}=Y_u(k)$, resulting in no privacy for any user.

We want to make certain that the adversary obtains no information about $X_1(k)$, the sample of data of user $1$ at time $k$. To do such, we will establish that there are a large number of users whom have a probability $p_u$ that when obfuscated could have resulted in a probability consistent with $p_1$. Consider asking whether another probability $p_2$ is sufficiently close enough to be confused with $p_1$ after obfuscation; in particular, we will look for $p_2$ such that, even if the adversary is given the obfuscated probabilities $V_{\Pi(1)}$ and $V_{\Pi(2)}$, he/she cannot associate these probabilities with $p_1$ and $p_2$. This requires that the distributions $Q_{1}$ and $Q_{2}$ of the obfuscated data of user $1$ and user $2$ have significant overlap; we explore this next.

Recall that $Q_u=P_u+ (1-2P_u)R_u$, and $R_u\sim Uniform [0, a_n]$. Thus, we know $Q_u {|} P_u=p_u$ has a uniform distribution with length $(1-2p_u)a_n$. Specifically,
\[Q_u\bigg{|}P_u=p_u \sim Uniform \left[p_u,p_u+(1-2p_u)a_n\right].\]
Figure \ref{fig:piqi_a} shows the distribution of $Q_u$ given $P_u=p_u$.

\begin{figure}[h]
	\centering
	\includegraphics[width = 0.8\linewidth]{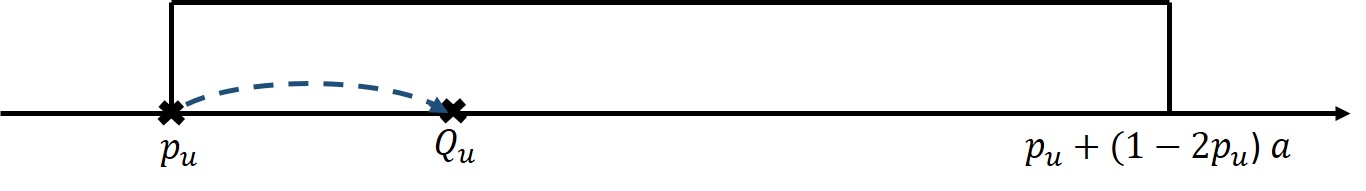}
	\caption{Distribution of $Q_u$ given $P_u=p_u$.}
	\label{fig:piqi_a}
\end{figure}

Consider two cases: In the first case, the support of the distributions $Q_1\big{|} P_1=p_1$ and $Q_2\big{|} P_2=p_2$ are small relative to the difference between $p_1$ and $p_2$ (Figure \ref{fig:case1}); in this case, given the probabilities $V_{\Pi(1)}$ and $V_{\Pi(2)}$ of the anonymized data sequences, the adversary can associate those with $p_1$ and $p_2$ without error. In the second case, the support of the distributions $Q_1\big{|} P_1=p_1$ and $Q_2\big{|} P_2=p_2$ is large relative to the difference between $p_1$ and $p_2$ (Figure \ref{fig:case2}), so it is difficult for the adversary to associate the probabilities $V_{\Pi(1)}$ and $V_{\Pi(2)}$ of the anonymized data sequences with $p_1$ and $p_2$. In particular, if $V_{\Pi(1)}$ and $V_{\Pi(2)}$ fall into the overlap of the support of $Q_1$ and $Q_2$, we will show the adversary can only guess randomly how to de-anonymize the data. Thus, if the ratio of the support of the distributions to $\big{|}p_1-p_2\big{|}$ goes to infinity, the adversary's posterior probability for each user converges to $\frac{1}{2}$, thus, implying no information leakage on the user identities. More generally, if we can guarantee that there will be a large set of users with $p_u$'s very close to $p_1$ compared to the support of $Q_1\big{|} P_1=p_1$, we will be able to obtain perfect privacy as demonstrated rigorously below.

\begin{figure}[h]
	\centering
	\includegraphics[width = 0.8\linewidth]{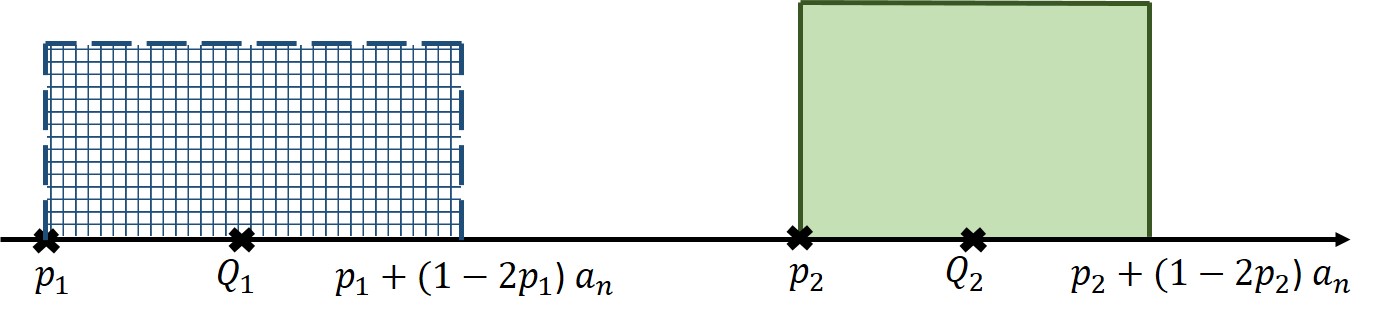}
	\caption{Case 1: The support of the distributions is small relative to the difference between $p_1$ and $p_2$.}
	\label{fig:case1}
\end{figure}

\begin{figure}[h]
	\centering
	\includegraphics[width = 0.8\linewidth]{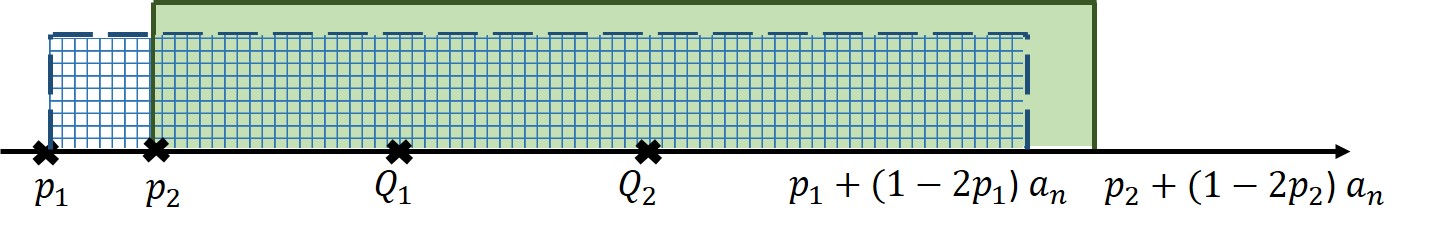}
	\caption{Case 2: The support of the distributions is large relative to the difference between $p_1$ and $p_2$.}
	\label{fig:case2}
\end{figure}

Given this intuition, the formal proof proceeds as follows. Given $p_1$, we define a set $J^{(n)}$ of users whose parameter $p_u$ of their data distributions is sufficiently close to $p_1$ (Figure \ref{fig:case2}; case 2), so that it is likely that $Q_1$ and $Q_u$ cannot be readily associated with $p_1$ and $p_u$.

The purpose of Lemmas \ref{lemOnePointFive}, \ref{lem2}, and \ref{lem3} is to show that, from the adversary's perspective, the users in set $J^{(n)}$ are indistinguishable. More specifically, the goal is to show that the obfuscated data corresponding to each of these users could have been generated by any other users in $J^{(n)}$ in an equally likely manner. To show this, Lemma \ref{lemOnePointFive} employs the fact that, if the observed values of $N$ uniformly distributed random variables ($N$ is size of set $J^{(n)}$) are within the intersection of their ranges, it is impossible to infer any information about the matching between the observed values and the distributions. That is, all possible $N!$ matchings are equally likely. Lemmas \ref{lem2} and \ref{lem3} leverage Lemma \ref{lemOnePointFive} to show that even if the adversary is given a set that includes all of the pseudonyms of the users in set $J^{(n)}$  (i.e., $\Pi(J^{(n)})\overset{\Delta}{=} \left\{\Pi^{-1}(u) \in J^{(n)}\right\}$) he/she still will not be able to infer any information about the matching of each specific user in set $J^{(n)}$ and his pseudonym. Then Lemma \ref{lem4} uses the above fact to show that the mutual information between the data set of user $1$ at time $k$ and the observed data sets of the adversary converges to zero for large enough $n$.


\noindent \textbf{Proof of Theorem \ref{two_state_thm}:}

\begin{proof}
Note, per Lemma~\ref{lemx} of Appendix \ref{sec:app_a}, it is sufficient to establish the results on a sequence of sets with high probability. That is, we can condition on high-probability events.

Now, define the critical set $J^{(n)}$ with size $N^{(n)}=\big{|}J^{(n)}\big{|}$ for $0\leq p_1<\frac{1}{2}$ as follows:
\[J^{(n)}=
\left\{  u \in \{1, 2, \dots, n\}: p_1 \leq P_u\leq p_1+\epsilon_n;  p_1+\epsilon_n\leq Q_u\leq p_1+(1-2p_1)a_n\right\},
\]
where $\epsilon_n \triangleq \frac{1}{n^{1-\frac{\beta}{2}}}$, $a_n= c'n^{-\left(1-\beta\right)}$ ,and $\beta$ is defined in the statement of Theorem \ref{two_state_thm}.

Note for large enough $n$, if $0\leq p_1<\frac{1}{2}$, we have $0\leq p_u<\frac{1}{2}$. As a result,
\[Q_u\bigg{|}P_u=p_u \sim Uniform \left(p_u,p_u+(1-2p_u)a_n\right).\]
We can prove that with high probability, $1 \in J^{(n)}$ for large enough $n$, as follows. First, Note that
\[Q_1\bigg{|}P_1=p_1 \sim Uniform \left(p_1,p_1+(1-2p_1)a_n\right).\]
Now, according to Figure \ref{fig:piqi_c},
	\begin{align}
	\no P\left(1 \in J^{(n)} \right) &= 1- \frac{\epsilon_n}{ \left(1-2p_1 \right)a_n}\\
	\nonumber &= 1- \frac{1}{ \left(1-2p_1 \right)c'n^{\frac{\beta}{2}}}, \ \
	\end{align}
thus, for any $c'>0$ and large enough $n$,
\begin{align}
	\no P\left(1 \in J^{(n)} \right) \to 1.
\end{align}
\begin{figure}[h]
	\centering
	\includegraphics[width = 0.8\linewidth]{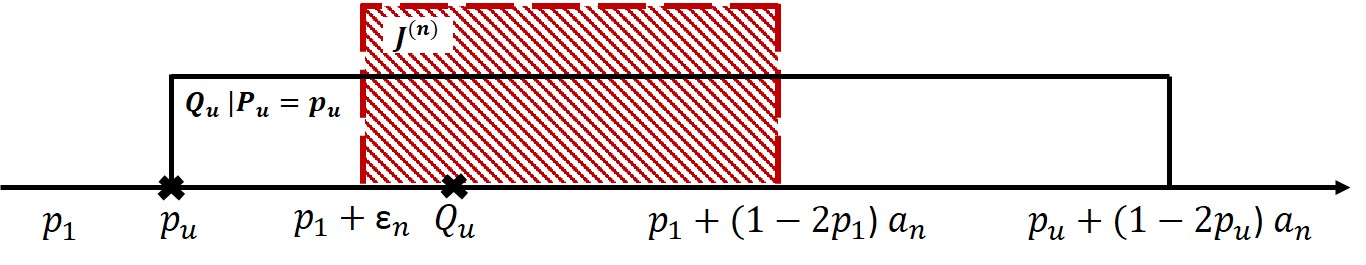}
	\caption{Range of $P_u$ and $Q_u$ for elements of set $J^{(n)}$ and probability density function of $Q_u\bigg{|}P_u=p_u$.}
	\label{fig:piqi_b}
\end{figure}

\begin{figure}[h]
	\centering
	\includegraphics[width = 0.8\linewidth]{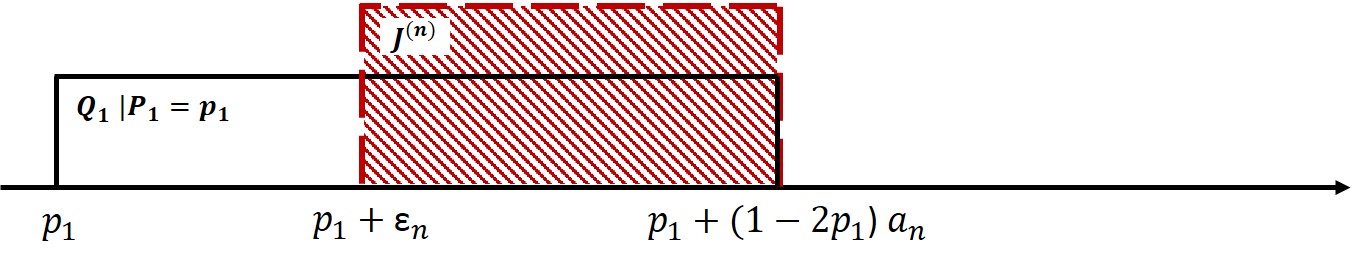}
	\caption{Range of $P_u$ and $Q_u$ for elements of set $J^{(n)}$ and probability density function of $Q_1\bigg{|}P_1=p_1$.}
	\label{fig:piqi_c}
\end{figure}

Now in the second step, we define the probability $W_j^{(n)}$ for any $j \in \Pi(J^{(n)})=\{\Pi(u): u \in J^{(n)} \}$ as
\[W_j^{(n)}= P\left(\Pi(1)=j \bigg{|} \textbf{V}, \Pi (J^{(n)})\right).\]
$W_j^{(n)}$ is the conditional probability that $\Pi(1)=j$ after perfectly observing the values of the permuted version of obfuscated probabilities ($\textbf{V}$) and set including all of the pseudonyms of the users in set $J^{(n)}$ $\left(\Pi(J^{(n)})\right)$. Since $\textbf{V}$ and $\Pi (J^{(n)})$ are random, $W_j^{(n)}$ is a random variable. However, we will prove shortly that in fact $W_j^{(n)}=\frac{1}{N^{(n)}}$, for all $j \in \Pi (J^{(n)})$.

Note: Since we are looking from the adversary's point of view, the assumption is that all the values of $P_u$, $u \in \{1,2,\cdots,n\}$ are known, so all of the probabilities are conditioned on the values of $P_1=p_1, P_2=p_2, \cdots, P_n=p_n$. Thus, to be accurate, we should write
\[W_j^{(n)}= P\left(\Pi(1)=j \bigg{|} \textbf{V}, \Pi (J^{(n)}), P_1, P_2, \cdots, P_n\right).\]
Nevertheless, for simplicity of notation, we often omit the conditioning on $P_1, P_2, \cdots, P_n$.

First, we need a lemma from elementary probability.

\begin{lem}
	\label{lemOnePointFive}
	Let $N$ be a positive integer, and let $a_1, a_2, \cdots, a_N$ and $b_1, b_2, \cdots, b_N$ be real numbers such that $a_u \leq b_u$ for all $u$. Assume that $X_1, X_2, \cdots, X_N$ are independent random variables such that
\[X_u \sim Uniform [a_u,b_u]. \]
Let also $\gamma_1, \gamma_2, \cdots, \gamma_N$ be distinct real numbers such that
\[ \gamma_j \in \bigcap_{u=1}^{N} [a_u, b_u] \ \ \textrm{for all }j \in \{1,2,..,N\}.\]
Suppose that we know the event $E$ has occurred, meaning that the observed values of $X_u$'s are equal to the set of $\gamma_j$'s (but with unknown ordering), i.e.,
\[E \ \ \equiv \ \ \{X_1, X_2, \cdots, X_N\}= \{ \gamma_1, \gamma_2, \cdots, \gamma_N \},\] then
\[P\left(X_1=\gamma_j |E\right)=\frac{1}{N}. \]
\end{lem}

\begin{proof}
Lemma \ref{lemOnePointFive} is proved in Appendix \ref{sec:app_b}.
\end{proof}



Using the above lemma, we can state our desired result for $W_j^{(n)}$.

\begin{lem}
	\label{lem2}
	For all $j \in \Pi (J^{(n)})$, $W_j^{(n)}=\frac{1}{N^{(n)}}.$
\end{lem}
\begin{proof}
	

We argue that the setting of this lemma is essentially equivalent to the assumptions in Lemma \ref{lemOnePointFive}. First, remember that
\[W_j^{(n)}= P\left(\Pi(1)=j \bigg{|} \textbf{V}, \Pi (J^{(n)})\right).\]

Note that ${Q}_u= P_u+(1-2P_u)R_u$, and since $R_u$ is uniformly distributed, ${Q}_u$ conditioned on $P_u$ is also uniformly distributed in the appropriate intervals. Moreover, since ${V}_{u} = {Q}_{\Pi^{-1}(u)}$, we conclude ${V}_{u}$ is also uniformly distributed. So, looking at the definition of $W_j^{(n)}$, we can say the following: given the values of the uniformly distributed random variables ${Q}_u$, we would like to know which one of the values in  $\textbf{V}$ is the actual value of ${Q}_1={V}_{\Pi(1)}$, i.e., is $\Pi(1)=j$? This is equivalent to the setting of Lemma \ref{lemOnePointFive} as described further below. 

Note that since $1 \in J^{(n)}$, $\Pi(1) \in \Pi (J^{(n)})$. Therefore, when searching for the value of $\Pi(1)$, it is sufficient to look inside set $\Pi (J^{(n)})$. Therefore, instead of looking among all the values of ${V}_{j}$, it is sufficient to look at ${V}_{j}$ for $j \in  \Pi (J^{(n)})$. Let's show these values by $\textbf{V}_{\Pi} =\{v_1, v_2, \cdots, v_{N^{(n)}} \}$, so,
\[W_j^{(n)}= P\left(\Pi(1)=j \bigg{|} \textbf{V}_{\Pi}, \Pi (J^{(n)})\right).\]

Thus, we have the following scenario: $Q_u, u \in  J^{(n)}$ are independent random variables, and
\[Q_u\big{|} P_u=p_u \sim Uniform [p_u, p_u+(1-2p_u)a_n]. \]
Also, $v_1, v_2, \cdots, v_{N^{(n)}}$ are the observed values of $Q_u$ with unknown ordering (unknown mapping $\Pi$). We also know from the definition of set $J^{(n)}$ that
\[P_u \leq p_1+\epsilon_n \leq Q_u,\]
\[Q_u \leq p_1(1-2a_n)+an \leq P_u(1-2a_n)+a_n,\]
so, we can conclude
\[ v_j \in \bigcap_{u=1}^{N^{(n)}} [p_u, p_u+(1-2p_u)a_n] \ \ \textrm{for all }j \in \{1,2,..,N^{(n)}\}. \]
We know the event $E$ has occurred, meaning that the observed values of $Q_u$'s are equal to set of $v_j$'s (but with unknown ordering), i.e.,
\[E \ \ \equiv \ \ \{Q_u, u \in  J^{(n)}\}= \{ v_1, v_2, \cdots, v_{N^{(n)}} \}. \]
Then, according to Lemma \ref{lemOnePointFive},
\[P\left(Q_1=v_j |E, P_1, P_2, \cdots, P_n \right)=\frac{1}{N^{(n)}}. \]
Note that there is a subtle difference between this lemma and Lemma \ref{lemOnePointFive}. Here $N^{(n)}$ is a random variable while $N$ is a fixed number in Lemma \ref{lemOnePointFive}. Nevertheless, since the assertion holds for every fixed $N$, it also holds for the case where $N$ is a random variable. Now, note that
\begin{align*}
P\left(Q_1=v_j |E, P_1, P_2, \cdots, P_n \right) &= P\left(\Pi(1)=j \bigg{|} E, P_1, P_2, \cdots, P_n \right)\\
&=P\left(\Pi(1)=j \bigg{|} \textbf{V}_{\Pi}, \Pi (J^{(n)}), P_1, P_2, \cdots, P_n \right)\\
&=W_j^{(n)}.
\end{align*}
Thus, we can conclude
		\[W_j^{(n)}=\frac{1}{N^{(n)}}.\]
\end{proof}	
%
%
In the third step, we define $\widetilde{W_j^{(n)}}$ for any $j \in \Pi (J^{(n)})$ as
\[\widetilde{W_j^{(n)}} = P\left(\Pi(1)=j \bigg{|} \textbf{Y}, \Pi (J^{(n)})\right).\]

$\widetilde{W_j^{(n)}}$ is the conditional probability that $\Pi(1)=j$ after observing the values of the anonymized version of the obfuscated samples of the users' data ($\textbf{Y}$) and the aggregate set including all the pseudonyms of the users in set $J^{(n)}$ (i.e., $\Pi(J^{(n)})\overset{\Delta}{=} \left\{\Pi^{-1}(j) \in J^{(n)}\right\}$). Since $\textbf{Y}$ and $\Pi (J^{(n)})$ are random, $\widetilde{W_j^{(n)}} $ is a random variable. Now, in the following lemma, we will prove $\widetilde{W_j^{(n)}} =\frac{1}{N^{(n)}}$, for all $j \in \Pi (J^{(n)})$ by using Lemma \ref{lem3}.

Note in the following lemma, we want to show that even if the adversary is given a set including all of the pseudonyms of the users in set $J^{(n)}$, he/she cannot match each specific user in set $J^{(n)}$ and his pseudonym.

\begin{lem}
	\label{lem3}
	For all $j \in \Pi (J^{(n)})$, $\widetilde{W_j^{(n)}}=\frac{1}{N^{(n)}}.$
\end{lem}

\begin{proof}
First, note that


\begin{align*}
\widetilde{W_j^{(n)}} = \sum_{\text{for all v}} P\left(\Pi(1)=j \bigg{|} \textbf{Y}, \Pi \left(J^{(n)}\right), \textbf{V}=\textbf{v}\right) P\left(\textbf{V}=\textbf{v} \bigg{|} \textbf{Y}, \Pi \left(J^{(n)}\right)\right).
\end{align*}

Also, we note that given $\textbf{V}$, $\Pi(J^{(n)})$, and $\textbf{Y}$ are independent. Intuitively, this is because when observing $\textbf{Y}$, any information regarding $\Pi(J^{(n)})$ is leaked through estimating $\textbf{V}$. This can be rigorously proved similar to the proof of Lemma 1 in \cite{tifs2016}. We can state this fact as
\[
	P\left(Y_u(k)\  \bigg{ | }  \ {V}_u=v_u, \Pi(J^{(n)}) \right)  = P\left(Y_u(k)\  \bigg{ | } \ {V}_u=v_u\right)=v_u.
\]
The right and left hand side are given by $Bernoulli (v_u)$ distributions.

As a result,
\[
\widetilde{W_j^{(n)}} = \sum_{\text{for all \textbf{v}}} P\left(\Pi(1)=j \bigg{|} \Pi (J^{(n)}), \textbf{V}=\textbf{v}\right) P\left(\textbf{V}=\textbf{v} \bigg{|} \textbf{Y}, \Pi \left(J^{(n)}\right)\right).
\]
Note $W_j^{(n)}=P\left(\Pi(1)=j \bigg{|} \Pi (J^{(n)}), \textbf{V}\right)$, so
\begin{align}
\no \widetilde{W_j^{(n)}} &= \sum_{\text{for all \textbf{v}}}W_j^{(n)}  P\left(\textbf{V}=\textbf{v} \bigg{|} \textbf{Y}, \Pi \left(J^{(n)}\right) \right) \\
\nonumber &= \frac{1}{N^{(n)}}\sum_{\text{for all \textbf{v}}} P\left(\textbf{V}=\textbf{v}\bigg{|} \textbf{Y}, \Pi \left(J^{(n)}\right)\right) \\
\nonumber &= \frac {1}{N^{(n)}}.\ \
\end{align}
\end{proof}

To show that no information is leaked, we need to show that the size of set $J^{(n)}$ goes to infinity. This is established in Lemma \ref{lem1}.

\begin{lem}
	\label{lem1}
	If $N^{(n)} \triangleq |J^{(n)}| $, then $N^{(n)} \rightarrow \infty$ with high probability as $n \rightarrow \infty$.  More specifically, there exists $\lambda>0$ such that
\[
	P\left(N^{(n)} > \frac{\lambda}{2}n^{\frac{\beta}{2}}\right) \rightarrow 1.
	\]
\end{lem}

\begin{proof}
Lemma \ref{lem1} is proved in Appendix \ref{sec:app_c}.
\end{proof}

In the final step, we define $\widehat{W_j^{(n)}}$ for any $j \in \Pi (J^{(n)})$ as
\[\widehat{W_j^{(n)}}=P\left(X_1(k)=1 \bigg{|} \textbf{Y}, \Pi (J^{(n)})\right).\]
$\widehat{W_j^{(n)}}$ is the conditional probability that $X_1(k)=1$ after observing the values of the anonymized version of the obfuscated samples of the users' data ($\textbf{Y}$) and the aggregate set including all of the pseudonyms of the users in set $J^{(n)}$ ($\Pi (J^{(n)})$). $\widehat{W_j^{(n)}} $ is a random variable because $\textbf{Y}$ and $\Pi (J^{(n)})$ are random. Now, in the following lemma, we will prove $\widehat{W_j^{(n)}}$ converges in distribution to $p_1$.

Note that this is the probability from the adversary's point of view. That is, given that the adversary has observed $\textbf{Y}$ as well as the extra information $ \Pi (J^{(n)})$, what can he/she infer about $X_1(k)$?
\begin{lem}
	\label{lem4}
	For all $j \in \Pi (J^{(n)})$, $\widehat{W_j^{(n)}} \xrightarrow{d} p_1.$
\end{lem}

\begin{proof}
We know
\begin{align*}
\widehat{W_j^{(n)}}= \sum_{j \in \Pi(J^{(n)})} P\left(X_1(k)=1 \bigg{|} \Pi(1)=j, \textbf{Y}, \Pi (J^{(n)})\right) P\left(\Pi(1)=j \bigg{|} \textbf{Y}, \Pi (J^{(n)})\right),
\end{align*}
and according to the definition $\widetilde{W_j^{(n)}}=P \left(\Pi(1)=j \bigg{|} \textbf{Y}, \Pi (J^{(n)})\right)$, we have
	\begin{align}
	\no \widehat{W_j^{(n)}} &= \sum_{j \in \Pi(J^{(n)})} P\left(X_1(k)=1 \bigg{|} \Pi(1)=j, \textbf{Y}, \Pi (J^{(n)})\right) \widetilde{W_j^{(n)}}\\
	\nonumber &= \frac{1}{N^{(n)}} \sum_{j \in \Pi(J^{(n)})} P\left(X_1(k)=1 \bigg{|} \Pi(1)=j, \textbf{Y}, \Pi (J^{(n)})\right).
	\end{align}
We now claim that
\[
 P\left(X_1(k)=1 \bigg{|} \Pi(1)=j, \textbf{Y}, \Pi (J^{(n)})\right)=p_1+o(1).
\]
The reasoning goes as follows. Given $\Pi(1)=j$ and knowing $\textbf{Y}$, we know that

\[
Y_{\Pi(1)}(k)={Z}_{1}(k)=\begin{cases}
{X}_{1}(k), & \textrm{with probability } 1-R_1.\\
1-{X}_{1}(k), & \textrm{with probability } R_1.
\end{cases}
\]

Thus, given $Y_{j}(k)=1$, Bayes' rule yields:
\begin{align*}
 P\left(X_1(k)=1 \bigg{|} \Pi(1)=j, \textbf{Y}, \Pi (J^{(n)})\right)&= \left(1- R_1 \right) \frac{P(X_1(k)=1)}{P(Y_{\Pi(1)}(k)=1)}\\
 &=\left(1-R_1 \right) \frac{p_1}{p_1 (1- R_1)+(1-p_1)R_1}\\
 &=1-o(1),
\end{align*}
and similarly, given $Y_{j}(k)=0$,
\begin{align*}
P\left(X_1(k)=1 \bigg{|} \Pi(1)=j, \textbf{Y}, \Pi (J^{(n)})\right)&= R_1 \frac{P(X_1(k)=1)}{P(Y_{\Pi(1)}(k)=0)}\\
&=R_1 \frac{p_1}{p_1 (1- R_1)+(1-p_1)R_1}\\
&=o(1).
\end{align*}
Note that by the independence assumption, the above probabilities do not depend on the other values of $Y_{u}(k)$ (as we are conditioning on $\Pi(1)=j$ ).
Thus, we can write
\begin{align}
	\no \widehat{W_j^{(n)}} &= \frac{1}{N^{(n)}} \sum_{j \in \Pi(J^{(n)})} P\left(X_1(k)=1 \bigg{|} \Pi(1)=j, \textbf{Y}, \Pi (J^{(n)})\right)\\
\no &=\frac{1}{N^{(n)}} \sum_{j \in \Pi(J^{(n)}), Y_{j}(k)=1}  (1-o(1)) + \frac{1}{N^{(n)}} \sum_{j \in \Pi(J^{(n)}), Y_{j}(k)=0}  o(1).
\end{align}
First, note that since $\abs*{\left\{j \in \Pi(J^{(n)}), Y_{j}(k)=0\right\}} \leq N^{(n)}$, the second term above converges to zero, thus,
\begin{align}
	\no \widehat{W_j^{(n)}}  \rightarrow \frac{\abs*{\left\{ j \in \Pi(J^{(n)}), Y_{\Pi(1)}(k)=1\right\} }}{N^{(n)}}.
\end{align}
Since for all $j \in \Pi(J^{(n)})$, $ Y_{j}(k)  \sim Bernoulli \left(p_1+o(1)\right)$, by a simple application of Chebyshev's inequality, we can conclude $\widehat{W_j^{(n)}}\rightarrow p_1$. Appendix \ref{sec:app_d} provides the detail.
\end{proof}	

As a result,
\begin{align*}
X_1(k) {|} \textbf{Y}, \Pi (J^{(n)})\rightarrow \textit{Bernoulli} (p_1),
\end{align*}
thus,
\[H\left(X_1(k) \bigg{|} \textbf{Y}, \Pi (J^{(n)})\right)\rightarrow H\left(X_1(k)\right).\]
Since conditioning reduces entropy,
\[H\left(X_1(k) \bigg{|} \textbf{Y}, \Pi (J^{(n)})\right)\leq H\left(X_1(k) \bigg{|} \textbf{Y}\right),\]
and as a result,
\[\lim_{n \rightarrow \infty}   H\left(X_1(k)\right)-H\left(X_1(k) \bigg{|} \textbf{Y}\right) \leq 0,\]
and 
\[\lim_{n \rightarrow \infty}   I\left(X_1(k);\textbf{Y}\right)\leq 0.\]
By knowing that $I\left(X_1(k);\textbf{Y}\right)$ cannot take any negative value, we can conclude that
\[I\left(X_1(k);\textbf{Y}\right)\rightarrow 0.\]
\end{proof}

\subsection{Extension to $r$-States}
Now, assume users' data samples can have $r$ possibilities $\left(0, 1, \cdots, r-1\right)$, and $p_u(i)$ shows the probability of user $u$ having data sample $i$. We define the vector $\textbf{p}_u$ and the matrix $\textbf{p}$ as
\[\textbf{p}_u= \begin{bmatrix}
p_u(1) \\ p_u(2) \\ \vdots \\p_u(r-1) \end{bmatrix} , \ \ \  \textbf{p} =\left[ \textbf{p}_{1}, \textbf{p}_{2}, \cdots,  \textbf{p}_{n}\right].
\]
We assume $p_u(i)$'s are drawn independently from some continuous density function, $f_\textbf{P}(\textbf{p}_u)$, which has support on a subset of the $(0,1)^{r-1}$ hypercube (Note that the $p_u(i)$'s sum to one, so one of them can be considered as the dependent value and the dimension is $r-1$). In particular, define the range of the distribution as
\begin{align}
\no  \mathcal{R}_{\textbf{p}} &= \{ (x_1,x_2,  \cdots, x_{r-1}) \in (0,1)^{r-1}:  x_i > 0 , x_1+x_2+\cdots+x_{r-1} < 1,\ i=1, 2, \cdots, r-1\}.
\end{align}
Figure~\ref{fig:rp} shows the range $\mathcal{R}_{\textbf{p}}$ for the case where $r=3$.
\begin{figure}[h]
	\centering
	\includegraphics[width = 0.5\linewidth]{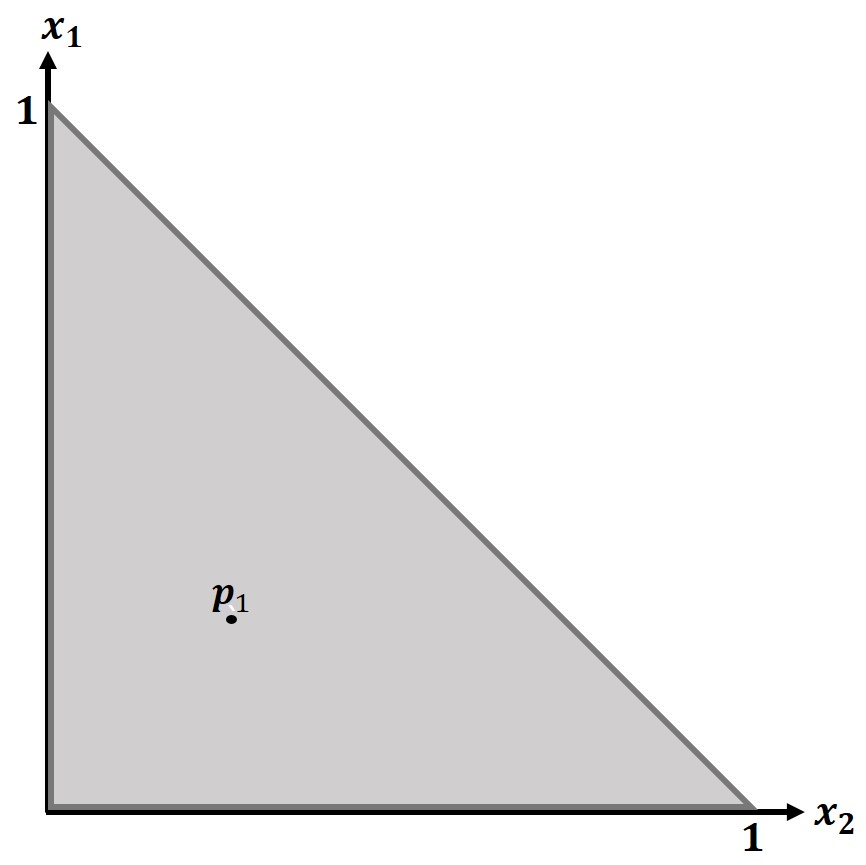}
	\caption{$\mathcal{R}_{\textbf{p}}$ for case $r=3$.}
	\label{fig:rp}
\end{figure}

Then, we assume there are $\delta_1, \delta_2>0$ such that:
\begin{equation}
\begin{cases}
\no    \delta_1<f_{\textbf{P}}(\mathbf{p}_u) <\delta_2, & \textbf{p}_u \in \mathcal{R}_{\textbf{p}}.\\
    f_{\textbf{P}}(\mathbf{p}_u)=0, &  \textbf{p}_u \notin \mathcal{R}_{\textbf{p}}.
\end{cases}
\end{equation}

The obfuscation is similar to the two-states case. Specifically, for $l \in \{0,1,\cdots, r-1\}$, we can write
\[
P({Z}_{u}(k)=l| X_{u}(k)=i) =\begin{cases}
1-R_u, & \textrm{for } l=i.\\
\frac{R_u}{r-1}, & \textrm{for } l \neq i.
\end{cases}
\]

\begin{thm}\label{r_state_thm}
For the above $r$-states model, if $\textbf{Z}$ is the obfuscated version of $\textbf{X}$, and $\textbf{Y}$ is the anonymized version of $\textbf{Z}$ as defined previously, and
\begin{itemize}
	 \item $m=m(n)$ is arbitrary;
	\item $R_u \sim Uniform [0, a_n]$, where $a_n \triangleq c'n^{-\left(\frac{1}{r-1}-\beta\right)}$ for any $c'>0$ and $0<\beta<\frac{1}{r-1}$;
\end{itemize}
then, user 1 has perfect privacy. That is,
\begin{align}
\no  \forall k\in \mathbb{N}, \ \ \ \lim\limits_{n\rightarrow \infty} I \left(X_1(k);{\textbf{Y}}\right) =0.
\end{align}\end{thm}

The proof of Theorem \ref{r_state_thm} is similar to the proof of Theorem \ref{two_state_thm}. The major difference is that instead of the random variables $P_u, Q_u, V_u$, we need to consider the random vectors $\textbf{P}_u, \textbf{Q}_u, \textbf{V}_u$.  Similarly, for user $u$, we define the vector $\textbf{Q}_u$ as
\[\textbf{Q}_u= \begin{bmatrix}
Q_u(1) \\ Q_u(2) \\ \vdots \\Q_u(r-1) \end{bmatrix}.
\]

In the $r$-states case,
\begin{align}
\no {Q}_u(i) &=P_u(i)\bigg(1-R_u(i) \bigg)+\bigg(1-P_u(i)\bigg)\frac{R_u}{r-1}  \\
\nonumber &= P_u+\bigg(1-r P_u\bigg)\frac{R_u}{r-1}.\ \
\end{align}
We also need to define the critical set $J^{(n)}$.  First, for $i=0,1, \cdots, r-1$, define set $J_i^{(n)}$ as follows. If $0\leq p_1(i)<\frac{1}{r}$, then,
\begin{align*}
&J_i^{(n)}= \\
&\left\{u \in \{1, 2, \dots, n\}: p_1(i) \leq P_u(i)\leq p_1(i)+\epsilon_n; p_1(i)+\epsilon_n\leq Q_u(i)\leq p_1(i)+(1-r p_1(i))\frac{a_n}{r-1}\right\},
\end{align*}
where $\epsilon_n \triangleq \frac{1}{n^{\frac{1}{r-1}-\frac{\beta}{2}}}$,  $a_n = c'n^{-\left(\frac{1}{r-1}-\beta\right)}$, and $\beta$ is defined in the statement of Theorem \ref{r_state_thm}.

We then define the critical set $J^{(n)}$ as:
\[
J^{(n)}=\bigcap_{l=0}^{r-1} J_i^{(n)}.
\]
We can then repeat the same arguments in the proof of Theorem \ref{two_state_thm} to complete the proof.


\section{Converse Results: No Privacy Region}
\label{converse}

In this section, we prove that if the number of observations by the adversary is larger than its critical value and the noise level is less than its critical value, then the adversary can find an algorithm to successfully estimate users' data samples with arbitrarily small error probability. Combined with the results of the previous section, this implies that asymptotically (as $n \rightarrow \infty$), privacy can be achieved \emph{if and only if} at least one of  the two techniques (obfuscation or anonymization) are used above their thresholds. This statement needs a clarification as follows:  Looking at the results of \cite{tifs2016}, we notice that anonymization alone can provide perfect privacy if $m(n)$ is below its threshold. On the other hand, the threshold for obfuscation requires some anonymization: In particular, the identities of the users must be permuted once to prevent the adversary from readily identifying the users.


\subsection{Two-States Model}
Again, we start with the i.i.d.\ two-states model. The data sample of user $u$ at any time is a Bernoulli random variable with parameter $p_u$.

As before, we assume that $p_u$'s are drawn independently from some continuous density function, $f_P(p_u)$, on the $(0,1)$ interval. Specifically, there are $\delta_1, \delta_2>0$ such that:
\begin{equation}
\no\begin{cases}
    \delta_1<f_P(p_u) <\delta_2, & p_u \in (0,1).\\
    f_P(p_u)=0, &  p _u\notin (0,1).
\end{cases}
\end{equation}

\begin{thm}\label{two_state_thm_converse}
For the above two-states mode, if $\textbf{Z}$ is the obfuscated version of $\textbf{X}$, and $\textbf{Y}$ is the anonymized version of $\textbf{Z}$ as defined, and
\begin{itemize}
	\item $m =cn^{2 +  \alpha}$ for any $c>0$ and $\alpha>0$;
	\item $R_u \sim Uniform [0, a_n]$, where $a_n \triangleq c'n^{-\left(1+\beta\right)}$ for any $c'>0$ and $\beta>\frac{\alpha}{4}$;
\end{itemize}
then, user $1$ has no privacy as $n$ goes to infinity.
\end{thm}

Since this is a converse result, we give an explicit detector at the adversary and show that it can be used by the adversary to recover the true data of user $1$.

\begin{proof}
The adversary first inverts the anonymization mapping $\Pi$ to obtain $Z_1(k)$, and then estimates the value of $X_1(k)$ from that. To invert the anonymization, the adversary calculates the empirical probability that each string is in state $1$ and then assigns the string with the empirical probability closest to $p_1$ to user 1.

\begin{figure}
	\centering
	\includegraphics[width=.8\linewidth]{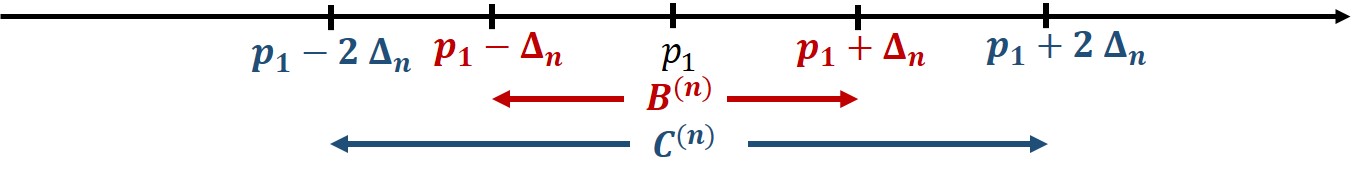}
	\centering
	\caption{$p_1$, sets $B^{(n)}$ and $C^{(n)}$ for case $r=2$.}
	\label{fig:converse}
\end{figure}

Formally, for $u=1, 2, \cdots, n$, the adversary computes $\overline{Y_u}$, the empirical probability of user $u$ being in state $1$, as follows:
\[
\overline{Y_u}=\frac{Y_u(1)+Y_u(2)+ \cdots +Y_u(m)}{m},
\]
thus,
\[
\overline{Y_{\Pi(u)}}=\frac{Z_u(1)+Z_u(2)+ \cdots +Z_u(m)}{m}.
\]

As shown in Figure \ref{fig:converse}, define
\[B^{(n)}\triangleq \left\{x \in (0,1); p_1-\Delta_n \leq x \leq p_1+\Delta_n\right\},\]
where $\Delta_n = \frac{1}{n^{1+\frac{\alpha}{4}}}$ and $\alpha $ is defined in the statement of Theorem \ref{two_state_thm_converse}. We claim that for $m =cn^{2 +  \alpha}$, $a_n=c'n^{-(1 +  \beta)}$, and large enough $n$,
\begin{enumerate}
\item $P\left( \overline{Y_{\Pi(1) }}\in B^{(n)}\right) \rightarrow 1.$
\item $P\left( \bigcup\limits_{u=2}^{n} \left(\overline{Y_{\Pi(u)}}\in B^{(n)}\right)\right) \rightarrow 0.$
\end{enumerate}
As a result, the adversary can identify $\Pi(1)$ by examining  $\overline{Y_u}$'s and assigning the one in $B^{(n)}$ to user $1$. Note that $\overline{Y_{\Pi(u) }} \in B^{(n)}$ is a set (event) in the underlying probability space and can be written as $\left\{\omega \in \Omega: \overline{Y_{\Pi(u) }}(\omega) \in B^{(n)}\right\}$.

First, we show that as $n$ goes to infinity,
\[P\left( \overline{Y_{\Pi(1) }}\in B^{(n)}\right) \rightarrow 1.\]
We can write
\begin{align}
\no P\left(\overline{Y_{\Pi(1)}} \in B^{(n)}\right) &= P\left(\frac{\sum\limits_{k=1}^{m}Z_1(k)}{m} \in  B^{(n)} \right)\\
\nonumber &= P\left(p_1-\Delta_n \leq\frac{\sum\limits_{k=1}^{m}Z_1(k)}{m}\leq p_1+\Delta_n \right)\\
\nonumber &= P\left(mp_1-m\Delta_n-mQ_1\leq \sum\limits_{k=1}^{m}Z_1(k)-mQ_1\leq  mp_1+m\Delta_n-mQ_1 \right).\ \
\end{align}
Note that for any $u \in \{1,2,\cdots, n \}$, we have
\begin{align}
\no |p_u-{Q}_u| &=|1-2p_u|R_u \\
\no & \leq R_u \leq a_n,
\end{align}
so, we can conclude
\begin{align}
\no P\left(\overline{Y_{\Pi(1)}} \in B^{(n)}\right) &= P\left(mp_1-m\Delta_n-mQ_1\leq \sum\limits_{k=1}^{m}Z_1(k)-mQ_1 \leq mp_1+m\Delta_n-mQ_1 \right)\\
\nonumber  &\geq P\left(-m\Delta_n+m a_n\leq \sum\limits_{k=1}^{m}Z_1(k)-mQ_1\leq -ma_n+m\Delta_n \right)\\
\nonumber &= P\left(\abs*{\sum\limits_{k=1}^{m}Z_1(k)-mQ_1}\leq m (\Delta_n-a_n) \right).\ \
\end{align}
Since $a_n \rightarrow 0$, for $p_1 \in (0,1)$ and large enough $n$, we can say $p_1+a_n < 2 p_1$. From Chernoff bound, for any $c,c',\alpha>0$ and $\beta>\frac{\alpha}{4}$,
\begin{align}
\no P\left(\abs*{\sum\limits_{k=1}^{m}Z_1(k)-mQ_1}\leq m (\Delta_n-a_n) \right) &\geq 1-2e^{-\frac{m(\Delta_n-a_n)^2}{3Q_1}} \\
\nonumber &\geq 1-2e^{-\frac{1}{3(p_1+a_n)}cn^{2+\alpha}\left(\frac{1}{n^{1+\frac{\alpha}{4}}}- \frac{c'}{n^{1 +  \beta}}\right)^2}\\
\nonumber &\geq 1-2e^{-\frac{c''}{6 p_1}n^{\frac{\alpha}{2}}}\rightarrow 1. \ \
\end{align}
As a result, as $n$ becomes large,
\[P\left(\overline{Y_{\Pi(1)}} \in B^{(n)}\right) \rightarrow 1.\]

Now, we need to show that as $n$ goes to infinity,
\[P\left( \bigcup\limits_{u=2}^{n} \left(\overline{Y_{\Pi(u)}}\in B^{(n)}\right)\right) \rightarrow 0.\]
First, we define
\[
C^{(n)}=\left\{x\in (0,1); p_1-2\Delta_n \leq x \leq p_1+2\Delta_n\right\}
,\]
and claim as $n$ goes to infinity,
\[P\left(\bigcup\limits_{u=2}^n \left(P_u \in C^{(n)}\right) \right)\rightarrow 0. \]

Note
\[4\Delta_n\delta_1< P\left( P_u\in C^{(n)}\right) < 4 \Delta_n\delta_2,\]
and according to the union bound, for large enough $n$,
\begin{align}
\no P\left( \bigcup\limits_{u=2}^n \left(P_u \in C^{(n)}\right) \right) &\leq \sum\limits_{u=2}^n P\left( P_u \in C^{(n)}\right) \\
\nonumber &\leq 4n \Delta_n \delta_2\\
\nonumber &= 4n \frac{1}{n^{1+{\frac{\alpha}{4}}}} \delta_2\\
\nonumber &= 4n^{-\frac{\alpha}{4}}\delta_2 \rightarrow 0. \ \
\end{align}
As a result, we can conclude that all $p_u$'s are outside of $C^{(n)}$ for $u \in \left\{2,3, \cdots, n\right\}$ with high probability.

Now, we claim that given all $p_u$'s are outside of $C^{(n)}$, $P\left(\overline{Y_{\Pi (u)}} \in B^{(n)}\right)$ is small. Remember that for any $u \in \{1,2,\cdots, n \}$, we have
\begin{align}
\no |p_u-{Q}_u| \leq a_n.
\end{align}
Now, noting the definitions of sets $B^{(n)}$ and $C^{(n)}$, we can write for $u \in \left\{2,3,  \cdots , n\right\}$,
\begin{align}
\nonumber
\no P\left(\overline{Y_{\Pi(u)}} \in B^{(n)}\right) &\leq P\left(\abs*{\overline{Y_{\Pi(u)}}-Q_u}\geq  (\Delta_n-a_n) \right)\\
\nonumber &= P\left(\abs*{\sum\limits_{k=1}^{m}Z_u(k)-mQ_u}> m(\Delta_n-a_n) \right).\ \
\end{align}
According to the Chernoff bound, for any $c,c',\alpha>0$ and $\beta>\frac{\alpha}{4}$,
\begin{align}
\no P\left(\abs*{\sum\limits_{k=1}^{m}Z_u(k)-mQ_u}> m(\Delta_n-a_n) \right) &\leq 2e^{-\frac{m(\Delta_n-a_n)^2}{3Q_1}} \\
\nonumber &\leq 2e^{-\frac{1}{3(p_1+a_n)}cn^{2+\alpha}\left(\frac{1}{n^{1+\frac{\alpha}{4}}}- \frac{c'}{n^{1 +  \beta}}\right)^2}\\
\nonumber &\leq 2e^{-\frac{c''}{6 p_1}n^{\frac{\alpha}{2}}}. \ \
\end{align}
Now, by using a union bound, we have
\begin{align}
\no P\left( \bigcup\limits_{u=2}^n \left(\overline{Y_{\Pi(u)}}\in B^{(n)}\right)\right)&\leq \sum\limits_{u=2}^{n}P\left(\overline{Y_{\Pi(u)}}\in B^{(n)}\right)\\
\nonumber &\leq n\left(2e^{-\frac{c''}{6 p_1}n^{\frac{\alpha}{2}}}\right),\ \
\end{align}
and thus, as $n$ goes to infinity,
\[P\left( \bigcup\limits_{u=2}^n \left(\overline{Y_{\Pi(u)}}\in B^{(n)}\right)\right) \rightarrow 0.\]

So, the adversary can successfully recover $Z_1(k)$. Since $Z_{1}(k)=X_1(k)$ with probability $1-R_1=1-o(1)$, the adversary can recover $X_{1}(k)$ with vanishing error probability for large enough $n$.
\end{proof}

\subsection{Extension to $r$-States}
Now, assume users' data samples can have $r$ possibilities $\left(0, 1, \cdots, r-1\right)$, and $p_u(i)$ shows the probability of user $u$ having data sample $i$. We define the vector $\textbf{p}_u$ and the matrix $\textbf{p}$ as
\[\textbf{p}_u= \begin{bmatrix}
p_u(1) \\ p_u(2) \\ \vdots \\p_u(r-1) \end{bmatrix} , \ \ \  \textbf{p} =\left[ \textbf{p}_{1}, \textbf{p}_{2}, \cdots,  \textbf{p}_{n}\right].
\]
We also assume $\textbf{p}_u$'s are drawn independently from some continuous density function, $f_P(\textbf{p}_u)$, which has support on a subset of the $(0,1)^{r-1}$ hypercube. In particular, define the range of distribution as
\begin{align}
\no  \mathcal{R}_{\textbf{p}} &= \left\{ (x_1, x_2, \cdots, x_{r-1}) \in (0,1)^{r-1}: x_i > 0 , x_1+ x_2+\cdots+ x_{r-1} < 1,\ \ i=1, 2,\cdots, r-1\right\}.
\end{align}
Then, we assume there are $\delta_1, \delta_2>0$ such that:
\begin{equation}
\begin{cases}
\no    \delta_1<f_{\textbf{P}}(\mathbf{p}_u) <\delta_2, & \textbf{p}_u \in \mathcal{R}_{\textbf{p}}.\\
    f_{\textbf{P}}(\mathbf{p}_u)=0, &  \textbf{p}_u \notin \mathcal{R}_{\textbf{p}}.
\end{cases}
\end{equation}

\begin{thm}\label{r_state_thm_converse}
For the above $r$-states mode, if $\textbf{Z}$ is the obfuscated version of $\textbf{X}$, and $\textbf{Y}$ is the anonymized version of $\textbf{Z}$ as defined, and
\begin{itemize}
	\item $m =cn^{\frac{2}{r-1} +  \alpha}$ for any $c>0$ and $0<\alpha<1$;
	\item $R_u \sim Uniform [0, a_n]$, where $a_n \triangleq c'n^{-\left(\frac{1}{r-1}+\beta\right)}$ for any $c'>0$ and $\beta>\frac{\alpha}{4}$;
\end{itemize}
then, user $1$ has no privacy as $n$ goes to infinity. 
\end{thm}

The proof of Theorem \ref{r_state_thm_converse} is similar to the proof of Theorem \ref{two_state_thm_converse}, so we just provide the general idea. We similarly define the empirical probability that the user with pseudonym $u$ has data sample $i$ $\left(\overline{{Y}_{u}}(i)\right)$ as follows:

\[
\overline{Y_u}(i)=\frac{\abs {\left\{k \in \{1, 2, \cdots, m\}:Y_u(k)=i\right\}}}{m},
\]
thus,
\[
\overline{Y_{\Pi(u)}}(i)=\frac{\abs {\left\{k \in \{1, 2, \cdots, m\}:Y_u(k)=i\right\}}}{m}.
\]

The difference is that now for each $u \in \{1,2,\cdots, n \}$, $\overline{\textbf{Y}_{u}}$ is a vector of size $r-1$. In other words,
\[\overline{\textbf{Y}_{u}}=\begin{bmatrix}
\overline{Y_u}(1) \\ \overline{Y_u}(2) \\ \vdots \\\overline{Y_u}(r-1) \end{bmatrix}.\]

\begin{figure}
  \centering
  \includegraphics[width=.5\linewidth, height=0.5 \linewidth]{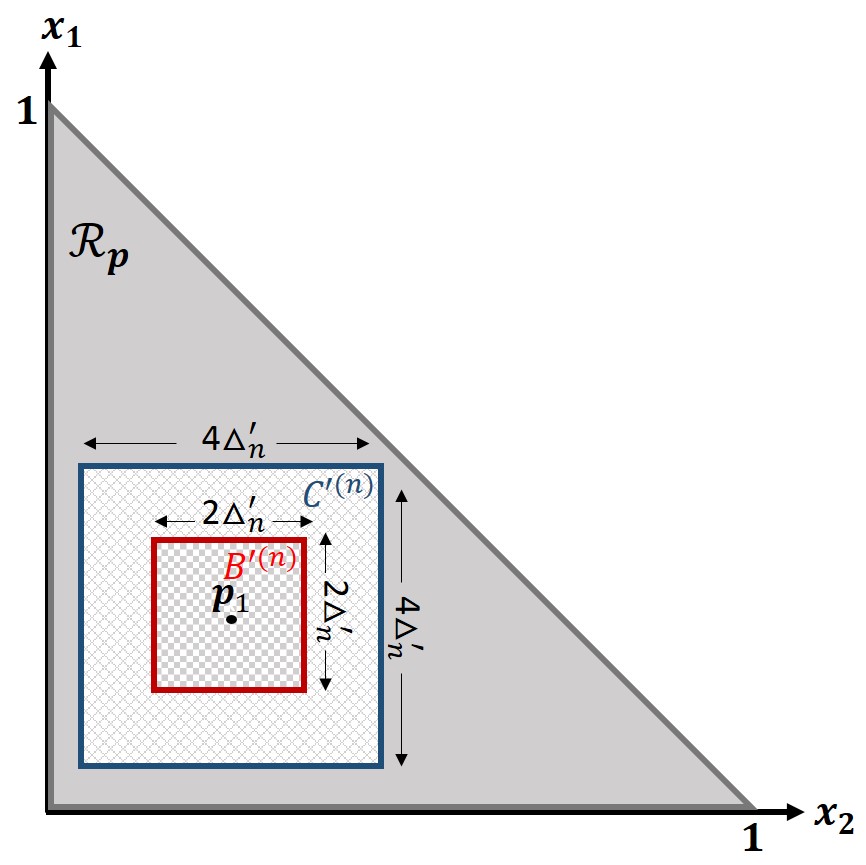}
  \centering
  \caption{$\textbf{p}_1$, sets $B'^{(n)}$ and $C'^{(n)}$ in $\mathcal{R}_\textbf{p}$ for case $r=3$.}
  \label{fig:rpp}
\end{figure}

Define sets $B'^{(n)}$ and $C'^{(n)}$ as
\begin{align}
 \no B'^{(n)}\triangleq & \left\{(x_1,x_2, \cdots ,x_{r-1}) \in \mathcal{R}_{\textbf{p}}: p_1(i)-\Delta'_n \leq x_i \leq p_1(i)+\Delta'_n,\ i=1,2, \cdots, r-1\right\},
\end{align}
\begin{align}
\no C'^{(n)}\triangleq &\left\{(x_1,x_2, \cdots ,x_{r-1}) \in \mathcal{R}_{\textbf{p}}: p_1(i)-2 \Delta'_n \leq x_i \leq p_1(i)+2 \Delta'_n,\ i=1,2, \cdots ,r-1\right\},
\end{align}
where $\Delta'_n = \frac{1}{n^{\frac{1}{r-1}+\frac{\alpha}{4}}}.$ Figure \ref{fig:rpp} shows $\textbf{p}_1$ and sets  $B'^{(n)}$ and $C'^{(n)}$ for the case $r=3.$

We claim for $m =cn^{\frac{2}{r-1} +  \alpha}$ and large enough $n$,
\begin{enumerate}
\item $P\left( \overline{\textbf{Y}_{\Pi(1) }}\in B'^{(n)}\right) \rightarrow 1$.
\item $P\left( \bigcup\limits_{u=2}^{n} \left(\overline{\textbf{Y}_{\Pi(u)}}\in B'^{(n)}\right)\right) \rightarrow 0.$
\end{enumerate}
The proof follows that for the two-states case. Thus, the adversary can de-anonymize the data and then recover $X_1(k)$ with vanishing error probability in the $r$-states model.

\subsection{Markov Chain Model} 
\label{subsec:markov}

So far, we have assumed users' data samples can have $r$ possibilities $\left(0, 1, \cdots, r-1\right)$ and users' pattern are i.i.d.\ . Here we model users' pattern using Markov chains to capture the dependency of the users' pattern over time. Again, we assume there are $r$ possibilities (the number of states in the Markov chains). Let $E$ be the set of edges. More specifically, $(i, l) \in E$ if there exists an edge from $i$ to $l$ with probability $ p(i,l)>0 $. What distinguishes different users is their transition probabilities $p_u(i,l)$ (the probability that user $u$ jumps from state $i$ to state $l$). The adversary knows the transition probabilities of all users. The model for obfuscation and anonymization is exactly the same as before.

We show that the adversary will be able to estimate the data samples of the users with low error probability if $m(n)$ and $a_n$ are in the appropriate range. The key idea is that the adversary can focus on a subset of the transition probabilities that are sufficient for recovering the entire transition probability matrix. By estimating those transition probabilities from the observed data and matching with the known transition probabilities of the users, the adversary will be able to first de-anonymize the data, and then estimate the actual samples of users' data. In particular, note that for each state $i$, we must have
\[\sum\limits_{l=1}^{r} p_u(i,l)=1,   \  \  \textrm{ for each }u \in \{1,2,\cdots, n \}, \]
so, the Markov chain of user $u$ is completely determined by a subset of size $d=|E|-r$ of transition probabilities. We define the vector $\textbf{p}_u$ and the matrix $\textbf{p}$ as
\[\textbf{p}_u= \begin{bmatrix}
p_u(1) \\ p_u(2) \\ \vdots \\p_u(|E|-r) \end{bmatrix} , \ \ \  \textbf{p} =\left[ \textbf{p}_{1}, \textbf{p}_{2}, \cdots,  \textbf{p}_{n}\right].
\]

We also consider $\textbf{p}_u$'s are drawn independently from some continuous density function, $f_P(\textbf{p}_u)$, which has support on a subset of the $(0,1)^{|E|-r}$ hypercube. Let $\mathcal{R}_{\textbf{p}} \subset \mathbb{R}^{d}$ be the range of acceptable values for $\textbf{p}_{u}$, so we have
\begin{align}
\no  \mathcal{R}_{\textbf{P}} &= \left\{ (x_1,x_2 \cdots, x_{d}) \in (0,1)^{d}: x_i > 0 , x_1+x_2+\cdots+x_{d} < 1,\ \ i=1,2,\cdots, d\right\}.
\end{align}
As before, we assume there are $ \delta_1, \delta_2 >0$, such that:
\begin{equation}
\begin{cases}
\no    \delta_1<f_{\textbf{P}}(\textbf{p}_u) <\delta_2, & \textbf{p}_u \in \mathcal{R}_{\textbf{p}}.\\
    f_{\textbf{P}}(\textbf{p}_u)=0, &  \textbf{p}_u \notin \mathcal{R}_{\textbf{p}}.
\end{cases}
\end{equation}

Using the above observations, we can establish the following theorem.

\begin{thm}\label{markov_thm}
For an irreducible, aperiodic Markov chain with $r$ states and $|E|$ edges as defined above, if $\textbf{Z}$ is the obfuscated version of $\textbf{X}$, and $\textbf{Y}$ is the anonymized version of $\textbf{Z}$, and
\begin{itemize}
 \item $m =cn^{\frac{2}{|E|-r} +  \alpha}$ for any $c>0$ and $\alpha>0$;
\item $R_u \sim Uniform [0, a_n]$, where $a_n \triangleq c'n^{-\left(\frac{1}{|E|-r}+\beta \right)}$ for any $c'>0$ and $\beta>\frac{\alpha}{4}$;
\end{itemize}
then, the adversary can successfully identify the data of user $1$ as $n$ goes to infinity. 
\end{thm}
The proof has a lot of similarity to the i.i.d.\ case, so we provide a sketch, mainly focusing on the differences. We argue as follows. If the total number of observations per user is $m=m(n)$, then define $M_i(u)$ to be the total number of visits by user $u$ to state $i$, for $i=0, 1, \cdots, r-1$. Since the Markov chain is irreducible and aperiodic, and $m(n) \rightarrow \infty$, all $\frac{M_i(u)}{m(n)}$ converge to their stationary values. Now conditioned on $M_i(u)=m_i(u)$, the transitions from state $i$ to state $l$ for user $u$ follow a multinomial distribution with probabilities $p_u(i,l)$.

Given the above, the setting is now very similar to the i.i.d.\ case. Each user is uniquely characterized by a vector $\textbf{p}_u$ of size $|E|-r$. We define sets $B^{''(n)}$ and $C^{''(n)}$ as
\[
 B^{''(n)}\triangleq \{(x_1, x_2, \cdots ,x_{d}) \in \mathcal{R}_{\textbf{p}}: p_1(i)-\Delta''_n \leq x_i \leq p_1(i)+\Delta''_n, i=1,2 , \cdots ,d\},
\]
\[
 C^{''(n)}\triangleq \{(x_1, x_2, \cdots ,x_{d}) \in \mathcal{R}_{\textbf{p}}:
  p_1(i)-2 \Delta''_n < x_i < p_1(i)+2 \Delta''_n, i=1,2, \cdots, d\},
\]
where $\Delta''_n= \frac{1}{n^{\frac{1}{|E|-r}+\frac{\alpha}{4}}}$, and $d= |E|-r$. Then, we can show that for the stated values of $m(n)$ and $a_n$, as $n$ becomes large:
\begin{enumerate}
\item $P\left( \overline{\textbf{Y}_{\Pi(1) }}\in B^{''(n)}\right) \rightarrow 1$,
\item $P\left( \bigcup\limits_{u=2}^{n} \left(\overline{\textbf{Y}_{\Pi(u)}}\in B{''}^{(n)}\right)\right) \rightarrow 0$,
\end{enumerate}
which means that the adversary can estimate the data of user $1$ with vanishing error probability. The proof is very similar to the proof of the i.i.d.\ case; however, there are two differences that need to be addressed:

First, the probability of observing an erroneous observation is not exactly given by $R_u$. In fact, a transition is distorted if at least one of its nodes is distorted. So, if the actual transition is from state $i$ to state $l$, then the probability of an erroneous observation is equal to
\begin{align}
	\no R'_u&=R_u+R_u-R_uR_u =R_u(2-R_u).
\end{align}

Nevertheless, here the order only matters, and the above expression is still in the order of $a_n =O \left( n^{-\left(\frac{1}{|E|-r}+\beta \right)} \right)$.

The second difference is more subtle. As opposed to the i.i.d.\ case, the error probabilities are not completely independent. In particular, if $X_u(k)$ is reported in error, then both the transition to that state and from that state are reported in error. This means that there is a dependency between errors of adjacent transitions. We can address this issue in the following way:  The adversary makes his decision only based on a subset of the observations. More specifically, the adversary looks at only odd-numbered transitions: First, third, fifth, etc., and ignores the even-numbered transitions.  In this way, the number of observations is effectively reduced from $m$ to $\frac{m}{2}$ which again does not impact the order of the result (recall that the Markov chain is aperiodic). However, the adversary now has access to observations with independent errors.

\section{Perfect Privacy Analysis: Markov Chain Model}\label{sec:perfect-MC}
So far, we have provided both achievability and converse results for the i.i.d.\ case. However, we have only provided the converse results for the Markov chain case. Here, we investigate achievability for Markov chain models. It turns out that for this case, the assumed obfuscation technique is not sufficient to achieve a reasonable level of privacy. Loosely speaking, we can state that if the adversary can make enough observations, then he can break the anonymity. The culprit is the fact that the sequence observed by the adversary is no longer modeled by a Markov chain; rather, it can be modeled by a hidden Markov chain. This allows the adversary to successfully estimate the obfuscation random variable $R_u$ as well as the $ p_u(i,l)$ values for each sequence, and hence successfully de-anonymize the sequences.

More specifically, as we will see below, there is a fundamental difference between the i.i.d.\ case and the Markov chain case. In the i.i.d.\ case, if the noise level is beyond a relatively small threshold, the adversary will be unable to de-anonymize the data and unable to recover the actual values of the data sets for users, \emph{regardless of the (large) size of $m=m(n)$}. On the other hand, in the Markov chain case, if $m=m(n)$ is large enough, then the adversary can easily de-anonymize the data. To better illustrate this, let's consider a simple example.

\begin{example}
Consider the scenario where there are only two states and the users' data samples change between the two states according to the Markov chain shown in Figure \ref{fig:MC-diagram}. What distinguishes the users is their different values of $p$. Now, suppose we use the same obfuscation method as before. That is, to create a noisy version of the sequences of data samples, for each user $u$, we generate the random variable $R_u$ that is the probability that the data sample of the user is changed to a different data sample by obfuscation. Specifically,
\[
{Z}_{u}(k)=\begin{cases}
{X}_{u}(k), & \textrm{with probability } 1-R_u.\\
1-{X}_{u}(k),& \textrm{with probability } R_u.
\end{cases}
\]

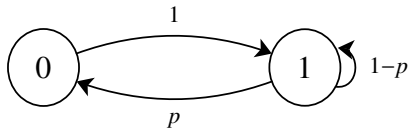
\begin{figure}[H]
\begin{center}
\[
\SelectTips {lu}{12 scaled 2500}
\xymatrixcolsep{6pc}\xymatrixrowsep{5pc}\xymatrix{
*++[o][F]{0}  \ar@/^1pc/[r]^{1}
& *++[o][F]{1} \ar@(dr,ur)[]_{1-p}  \ar@/^1pc/[l]^{p}
}
\]
\caption{A state transition diagram.}\label{fig:MC-diagram}
\end{center}
\end{figure}
To analyze this problem, we can construct the underlying Markov chain as follows. Each state in this Markov chain is identified by two values: the real state of the user, and the observed value by the adversary. In particular, we can write
\[\left(\text{Real value}, \text{Observed value}\right) \in \left\{\right(0,0), (0,1), (1,0), (1,1)\}.\]
Figure \ref{fig:MC-diagram2} shows the state transition diagram of this new Markov chain.

\begin{figure}[H]
\begin{center}
\[
\SelectTips {lu}{12 scaled 2000}
\xymatrixcolsep{10pc}\xymatrixrowsep{8pc}\xymatrix{
*++[o][F]{00} \ar@//[d]^{R} \ar@/_1pc/[dr]_>>>>>{1-R}
& *++[o][F]{01} \ar@//[d]^{1-R} \ar@/^1pc/[dl]^>>>>>{R} \\
*++[o][F]{10} \ar@(d,l)[]^{(1-p)R} \ar@/_2pc/[r]_{(1-p)(1-R)} \ar@/^2pc/[u]^{p(1-R)} \ar@/^1pc/[ur]^>>>>>{pR}
& *++[o][F]{11} \ar@(d,r)[]_{(1-p)(1-R)}  \ar@//[l]^{(1-p)R} \ar@/_2pc/[u]_{pR} \ar@/_1pc/[ul]_>>>>>{p(1-R)}
}
\]
\caption{The state transition diagram of the new Markov chain.}\label{fig:MC-diagram2}
\end{center}
\end{figure}
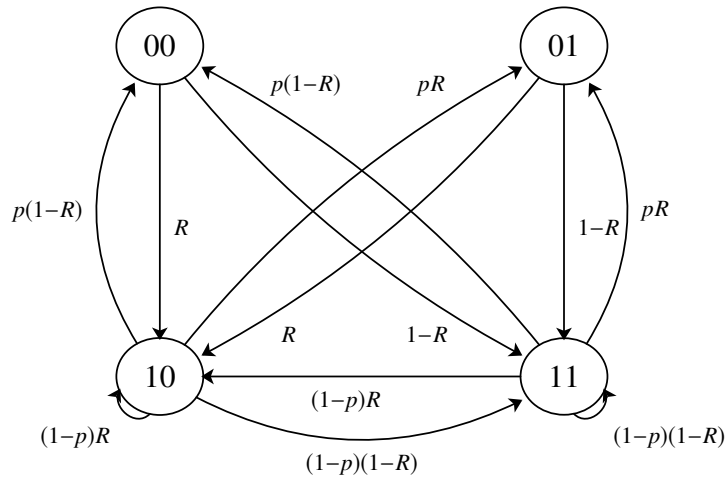
 We know
 \[ \pi_{00}=\pi_0(1-R)=\frac{p}{1+p}(1-R).\]
 \[ \pi_{01}=\pi_0R= \frac{p}{1+p}R.\]
 \[ \pi_{10}=\pi_1R=\frac{1}{1+p}R.\]
 \[ \pi_{11}=\pi_1(1-R)=\frac{1}{1+p}(1-R).\]

The observed process by the adversary is not a Markov chain; nevertheless, we can define limiting probabilities. In particular, let $\theta_0$ be the limiting probability of observing a zero. That is, we have
\[
\frac{M_0}{m} \xrightarrow{d} \theta_0,  \ \ \textrm{ as }n \rightarrow \infty,
\]
where $m$ is the total number of observations by the adversary, and $M_0$ is the number of $0$'s observed. Then,
\[\theta_0= \pi_{00}+\pi_{10} =\frac{(1-R)p+R}{1+p}.\]
Also, let $\theta_1$ be the limiting probability of observing a one, so
\[\theta_1= \pi_{01}+\pi_{11} =\frac{pR+(1-R)}{1+p}=1-\theta_0.\]

Now the adversary's estimate of $\theta_0$ is given by:
\begin{align}\label{eq1}
\hat{\theta}_0= \frac{(1-R)p+R}{1+p}.
\end{align}
Note that if the number of observations by the adversary can be arbitrarily large, the adversary can obtain an arbitrarily accurate estimate of $\theta_0$.
The adversary can obtain another equation easily, as follows.   Let $\theta_{01}$ be the limiting value of the portion of transitions from state $0$ to $1$ in the chain observed by the adversary. We can write
\begin{align}
\no \theta_{01} &=P \left\{(00\rightarrow 01), (00\rightarrow 11), (10 \rightarrow 01), (10 \rightarrow 11) \right\}\\\
\nonumber &= \pi_{00}(1-R)+\pi_{10}PR+ \pi_{10}(1-p)(1-R).\ \
\end{align}
As a result,
\begin{align}\label{eq2}
\hat{\theta}_{01}= \frac{p(1-R)^2+R\left(PR(1-R)(1-p)\right)}{1+p}.
\end{align}
Again, if the number of observations can be arbitrarily large, the adversary can obtain an arbitrarily accurate estimate of $\theta_{01}$.
By solving the Equations \ref{eq1} and \ref{eq2}, the adversary can successfully recover $R$ and $p$; thus, he/she can successfully determine the users' data values.

\end{example}

%% file: conclusion.tex
\section{Discussion}
\label{discussion}
\subsection{Markov Chain Model}
\label{markov}
As opposed to the i.i.d.\ case, we see from Section \ref{sec:perfect-MC} that if we do not limit $m=m(n)$, the assumed obfuscation method will not be sufficient to achieve perfect privacy. There are a few natural questions here. First, for a given noise level, what would be the maximum $m(n)$ that could guarantee perfect privacy in this model? The more interesting question is, how can we possibly modify the obfuscation technique to make it more suitable for the Markov chain model? A natural solution seems to be re-generating the obfuscation random variables $R_u$ periodically. This will keep the adversary from easily estimating them by observing a long sequence of data at a small increase in complexity. In fact, this will make the obfuscation much more \emph{robust} to modeling uncertainties and errors. It is worth noting, however, that this change would not affect the other results in the paper. That is, even if the obfuscation random variables are re-generated frequently, it is relatively easy to check that all the previous theorems in the paper remain valid. However, the increase in robustness to modeling errors will definitely be a significant advantage. Thus, the question is how often should the random variable $R_u$ be re-generated to strike a good balance between complexity and privacy? These are all interesting questions for future research.

\subsection{Obfuscating the Samples of Users' Data Using Continuous Noise}
Here we argue that for the setting of this paper, continuous noise such as that drawn from a Gaussian distribution is not a good option to obfuscate the sample of users' data drawn from a finite alphabet when we want to achieve perfect privacy. For a better understanding, let's consider a simple example.
\begin{example}
Consider the scenario where the users' datasets are governed by an i.i.d.\ model and the number of possible values for each sample of the users' data ($r$) is equal to 2 (two-states model). Note that the data sequence for user $u$ is a Bernoulli random variable with parameter $p_u$.

Assume that the actual sample of the data of user $u$ at time $k$ ($X_u(k)$) is obfuscated using noise drawn from a Gaussian distribution ($S_u(k)$), and $Z_u(k)$ is the obfuscated version of $X_u(k)$.  That is, we can write
\[Z_u(k)=X_u(k)+S_u(k); \ \ \ \ \  S_u(k) \sim N\left(0, R_u\right), \]
where $R_u$ is chosen from some distribution. For simplicity, we can consider $R_u\sim N\left(0, a^2_n\right)$ where $a_n$ is the noise level.

We also apply anonymization to $Z_u(k)$, and, as before, $Y_u(k)$ is the reported sample of the data of user $u$ at time $k$ after applying anonymization.  Per Section \ref{sec:framework}, anonymization is modeled by a random permutation $\Pi(u)$ on the set of $n$ users.

Now, the question is as follows: Is it possible to achieve perfect privacy independent of the number of adversary's observation ($m$) while using this continuous noise ($S_u(k)$) to obfuscate the sample of users' data?

Note that the density function of the reported sample of the data of user $u$ after applying obfuscation is
\begin{align}
\no f_{Z_u}(z)&= p_u f_{S_u(k)}(z-1)+(1-p_u) f_{S_u(k)}(z) \\
\nonumber &= p_u \frac{1}{\sqrt{2\pi}R_u}e^{-\frac{(z-1)^2}{2R_u}}+(1-p_u)\frac{1}{\sqrt{2\pi}R_u}e^{-\frac{z^2}{2R_u}}.\ \
\end{align}
In this case, when the adversary's number of observations is large, the adversary can estimate the values of $P_u$ and $R_u$ for each user with an arbitrarily small error probability. As a result, the adversary can de-anonymize the data and then recover $X_u(k)$. The conclusion here is that a continuous noise distribution gives too much information to the adversary when used for obfuscation of finite alphabet data. A method to remedy this issue is to regenerate the random variables $R_u$ frequently (similar to our previous discussion for Markov chains). Understanding the optimal frequency of such a regeneration and detailed analysis in this case is an interesting future research direction.
\end{example}


\subsection{Relation to Differential Privacy}
Differential privacy is mainly used when there is a statistical database of users' sensitive information, and the goal is to protect an individual's data while publishing aggregate information about the database \cite{ lee2012differential, bordenabe2014optimal, chatzikokolakis2015geo, nguyen2013differential, machanavajjhala2008privacy, kousha2}. The goal of differential privacy is publishing aggregate queries with low sensitivity, which means the effect of changes in a single individual on the outcome of the aggregated information is negligible.

In \cite{geo2013} three different approaches for differential privacy are presented. The one that best matches our setting is stated as
\begin{align*}
\frac{P(X_1(k)=x_1|\textbf{Y})}{P(X_1(k)=x_2|\textbf{Y})} \leq e^{\epsilon_r}\frac{P(X_1(k)=x_1)}{P(X_1(k)=x_2)},
\end{align*}
where $\textbf{Y}$ is the set of reported datasets.
It means $\textbf{Y}$ has a limited effect on the probabilities assigned by the attacker. In differential privacy, user $1$ has strongest differential privacy when $\epsilon_r=0$.

In Lemma \ref{lem4}, we proved that if user $1$ has perfect privacy, this implies that asymptotically (for large enough $n$)
\begin{align}\label{eq1}
P\left(X_1(k)=x_1 \big{|} \textbf{Y}\right)\rightarrow P\left(X_1(k)=x_1\right).
\end{align}
\begin{align}\label{eq2}
P\left(X_1(k)=x_2 \big{|} \textbf{Y}\right)\rightarrow P\left(X_1(k)=x_2\right).
\end{align}

As a result, by using (\ref{eq1}) and (\ref{eq2}), we can conclude that if we satisfy the perfect privacy condition given in this paper, we also satisfy differential privacy with $\epsilon_r=0$, i.e., the strongest case of differential privacy.

\section{Conclusions}
In this paper, we have considered both obfuscation and anonymization techniques to achieve privacy.  The privacy level of the users depends on both $m(n)$ (number of observations per user by the adversary for a fixed anonymization mapping) and $a_n$ (noise level). That is, larger $m(n)$ and smaller $a_n$ indicate weaker privacy.  We characterized the limits of privacy in the entire $m(n)-a_n$ plane for the i.i.d.\ case; that is, we obtained the exact values of the thresholds for $m(n)$ and $a_n$ required for privacy to be maintained.  We showed that if $m(n)$ is fewer than $O\left(n^{\frac{2}{r-1}}\right)$, or $a_n$ is larger than $\Omega\left(n^{-\frac{1}{r-1}}\right)$, users have perfect privacy. On the other hand, if neither of these two conditions is satisfied, users have no privacy. For the case where the users' patterns are modeled by Markov chains, we obtained a no-privacy region in the $m(n)-a_n$ plane.

Future research in this area needs to characterize the exact privacy/no-privacy regions when user data sequences obey Markov models. It is also important to consider different ways to obfuscate users' data sets and study the utility-privacy trade-offs for different types of obfuscation techniques.

%% file: appendix_a.tex
\section{Lemma \ref{lemx} and its Proof}
\label{sec:app_a}
Here we state that we can condition on high-probability events.

\begin{lem}
	\label{lemx}
	Let $p \in (0,1)$, and $X \sim Bernoulli (p)$ be defined on a probability space $(\Omega, \mathcal{F}, P)$. Consider $B_1, B_2, \cdots $ be a sequence of events defined on the same probability space such that $P(B_n) \rightarrow 1$ as $n$ goes to infinity. Also, let $\textbf{Y}$ be a random vector (matrix) in the same probability space, then:
	\[I(X; \textbf{Y}) \rightarrow 0\ \ \text{iff}\ \  I(X; \textbf{Y} {|} B_n) \rightarrow 0. \]
\end{lem}

\begin{proof}
First, we prove that as n becomes large,
\begin{align}\label{eq:H1}
H(X {|} B_n)- H(X) \rightarrow 0.
\end{align}

Note that as $n$ goes to infinity,
\begin{align}
\no P\left(X=1\right) &=P\left(X=1 \bigg{|} B_n\right) P\left(B_n\right) + P\left(X=1 \bigg{|} \overline{B_n}\right) P\left(\overline{B_n}\right)\\
\no &=P\left(X=1 \bigg{|} B_n\right),\ \
\end{align}
thus,
$\left(X \bigg{|} B_n\right) \xrightarrow{d} X$, and as $n$ goes to infinity,
 \[H\left(X {|} B_n\right)- H(X) \rightarrow 0.\]
Similarly, as $n$ becomes large,
\[ P\left(X=1 \bigg{|} \textbf{Y}=\textbf{y}\right) \rightarrow P\left(X=1 \bigg{|} \textbf{Y}=\textbf{y}, B_n\right),\ \
\]
and
\begin{align}\label{eq:H2}
H\left(X {|}  \textbf{Y}=\textbf{y},B_n\right)- H\left(X {|} \textbf{Y}=\textbf{y}\right) \rightarrow 0.
\end{align}
Remembering that
\begin{align}\label{eq:H3}
I\left(X; \textbf{Y}\right)=H(X)-H(X {|} \textbf{Y}),
\end{align}
and using (\ref{eq:H1}), (\ref{eq:H2}), and (\ref{eq:H3}), we can conclude that as  $n$ goes to infinity,
\[I\left(X;\textbf{Y} {|} B_n\right) - I\left(X,\textbf{Y}\right) \rightarrow 0.\]
As a result, for large enough $n$,
\[I\left(X; \textbf{Y}\right) \rightarrow 0 \Longleftrightarrow I\left(X; \textbf{Y} {|} B_n\right) \rightarrow 0. \]
\end{proof}

%% file: appendix_c.tex
\section{Proof of Lemma \ref{lemOnePointFive}}
\label{sec:app_b}
Here we provide a formal proof for Lemma \ref{lemOnePointFive} which we restate as follows.

Let $N$ be a positive integer, and let $a_1, a_2, \cdots, a_N$ and $b_1, b_2, \cdots, b_N$ be real numbers such that $a_u \leq b_u$ for all $u$. Assume that $X_1, X_2, \cdots, X_N$ are $N$ independent random variables such that
\[X_u \sim Uniform[a_u,b_u]. \]
Let also $\gamma_1, \gamma_2, \cdots, \gamma_N$ be real numbers such that
\[ \gamma_j \in \bigcap_{u=1}^{N} [a_u, b_u] \ \ \textrm{for all }j \in \{1,2,\cdots,N\}. \]
Suppose that we know the event $E$ has occurred, meaning that the observed values of $X_u$'s is equal to the set of $\gamma_j$'s (but with unknown ordering), i.e.,
\[E \ \ \equiv \ \ \{X_1,X_2,\cdots,X_N\}= \{ \gamma_1, \gamma_2, \cdots, \gamma_N \}, \] then
\[P\left(X_1=\gamma_j |E\right)=\frac{1}{N}. \]

\begin{proof}
Define sets $\mathfrak{P}$ and $\mathfrak{P}_j$ as follows:
\[\mathfrak{P}= \textrm{The set of all permutations $\Pi$ on }\{1,2,\cdots,N\}. \]
\[\mathfrak{P}_j= \textrm{The set of all permutations $\Pi$ on }\{1,2,\cdots,N\} \textrm{ such that } \Pi(1)=j. \]
We have $|\mathfrak{P}|=N!$ and $|\mathfrak{P}|=(N-1)!$. Then
\begin{align*}
 P(X_1=\alpha_j |E)&=\frac{\sum_{\pi \in \mathfrak{P}_j} f_{X_1,X_2, \cdots,X_N} (\gamma_{\pi(1)}, \gamma_{\pi(2)}, \cdots, \gamma_{\pi(N)})} {\sum_{\pi \in \mathfrak{P}} f_{X_1,X_2, \cdots,X_N} (\gamma_{\pi(1)}, \gamma_{\pi(2)}, \cdots, \gamma_{\pi(N)})}\\
 &=\frac{(N-1)! \prod\limits_{u=1}^{N} \frac{1}{b_u-a_u}}{N! \prod\limits_{u=1}^{N} \frac{1}{b_u-a_u}}\\
 &=\frac{1}{N}.
\end{align*}
\end{proof}

%% file: appendix_b.tex
\section{Proof of Lemma \ref{lem1}}
\label{sec:app_c}
Here, we provide a formal proof for Lemma \ref{lem1} which we restate as follows. The following lemma confirms that the number of elements in $J^{(n)}$ goes to infinity as $n$ becomes large.

	If $N^{(n)} \triangleq |J^{(n)}| $, then $N^{(n)} \rightarrow \infty$ with high probability as $n \rightarrow \infty$.  More specifically, there exists $\lambda>0$ such that
\[
	P\left(N^{(n)} > \frac{\lambda}{2}n^{\frac{\beta}{2}}\right) \rightarrow 1.
	\]

\begin{proof}
Define the events $A$, $B$ as
\[A \equiv  p_1\leq P_u\leq p_1+\epsilon_n\]
\[B \equiv p_1+\epsilon_n\leq Q_u\leq p_1+(1-2p_1)a_n.\]
Then, for $u \in \{1, 2, \dots, n\}$ and $0\leq p_1<\frac{1}{2}$:
	\begin{align}
	\no P\left(u\in J^{(n)}\right) &= P\left(A\ \cap \ B\right)\\
	\nonumber &= P\left(A\right) P\left(B \big{|}A \right). \ \
	\end{align}		
So, given $p_1 \in (0,1)$ and the assumption $0<\delta_1<f_p<\delta_2$, for $n$ large enough, we have
	\[
	P(A)  = \int_{p_1}^{ p_1+\epsilon_n}f_P(p) dp,
	\]
so, we can conclude that
	\[
	\epsilon_n\delta_1<P(A) <\epsilon_n\delta_2.
	\]
	We can find a $\delta$ such that $\delta_1<\delta<\delta_2$ and
\begin{equation}\label{eq:5}
	P( A) = \epsilon_n\delta.
\end{equation}
We know
	\[Q_u\bigg{|}P_u=p_u \sim Uniform \left[p_u,p_u+(1-2p_u)a_n\right],\]
	so, according to Figure \ref{fig:piqi_b}, for $p_1\leq p_u\leq p_1+\epsilon_n$,
	\begin{align}
	\no P\left(B | P_u=p_u\right ) &= \frac{p_1+(1-2p_1)a_n-p_1-\epsilon_n}{p_u+(1-2p_u)a_n-p_u} \\
	\nonumber &= \frac{(1-2p_1)a_n-\epsilon_n}{(1-2p_u)a_n}\\
	\nonumber &\geq \frac{(1-2p_1)a_n-\epsilon_n}{(1-2p_1)a_n} \\
	\nonumber &= 1- \frac{\epsilon_n}{(1-2p_1)a_n}, \ \
	\end{align}
which implies
\begin{align}
P\left(B | A\right ) \geq 1- \frac{\epsilon_n}{(1-2p_1)a_n}. \label{eq:6}
\end{align}
Using (\ref{eq:5}) and (\ref{eq:6}), we can conclude
	\[P\left(u\in J^{(n)}\right)\geq \epsilon_n\delta \left(1- \frac{\epsilon_n}{(1-2p_1)a_n}\right).\]
Then, we can say that $N^{(n)}$ has a binomial distribution with expected value of $N^{(n)}$ greater than $n\epsilon_n\delta \left(1- \frac{\epsilon_n}{(1-2p_1)a_n}\right)$, and by substituting $\epsilon_n$ and $a_n$, for any $c'>0$, we get
\[E\left[N^{(n)}\right] \geq \delta\left(n^{\frac{\beta}{2}}- \frac{1}{{c'(1-2p_1)}}\right)   \geq  \lambda n^{\frac{\beta}{2}}.\]

Now by using Chernoff bound, we have
\[P\left(N^{(n)} \leq (1- \theta) E\left[N^{(n)}\right]\right) \leq e^{-\frac{\theta^2}{2}E\left[N^{(n)}\right]},\]
so, if we assume $\theta=\frac{1}{2}$, we can conclude for large enough $n$,
\begin{align}
\nonumber P\left(N^{(n)} \leq \frac{\lambda}{2}n^{\frac{\beta}{2}}\right) &\leq P\left(N^{(n)} \leq \frac{E\left[N^{(n)}\right]}{2}\right)\\
\nonumber &\leq e^{-\frac{E[N^{(n)}]}{8}}\\
\nonumber &\leq e^{-\frac{\lambda n^{\frac{\beta}{2}}}{8}} \rightarrow 0.
\end{align}
As a result, $N^{(n)} \rightarrow \infty$ with high probability for large enough $n.$

\end{proof}

%% file: appendix_d.tex
\section{Completion of the Proof of Lemma \ref{lem4}}\label{sec:app_d}

Let $p_1 \in (0,1)$, and let $N^{(n)}$ be a random variable as above, i.e., $N^{(n)}  \rightarrow \infty$ as $n \rightarrow \infty$. Consider the sequence of independent random variables $Y_u \sim Bernoulli (p_u)$ for $u=1, 2, \cdots, N^{(n)}$ such that
\begin{enumerate}
	\item For all $n$ and all $u \in \left\{1, 2, \cdots, N^{(n)}\right\}$, $\abs*{p_u-p_1} \leq \zeta_n.$
	\item $\lim\limits_{n\to\infty} \zeta_n =0.$
\end{enumerate}
Define
\[\overline{Y}\triangleq \frac{1}{N^{(n)}}\sum_{u=1}^{N^{(n)}}Y_u,\]
then $\overline{Y} \xrightarrow{d} p_1.$

\begin{proof}
Note
\begin{align}
\no E[\overline{Y}] &= \frac{1}{N^{(n)}}\sum_{u=1}^{N^{(n)}}p_u\\
\nonumber &\leq \frac{1}{N^{(n)}} \sum_{u=1}^{N^{(n)}}\left(p_1+\zeta_n\right)\\
\nonumber &= \frac{1}{N^{(n)}}\cdot N^{(n)}(p_1+\zeta_n) \\
\nonumber &= p_1+\zeta_n. \ \
\end{align}

Similarly we can prove $E\left[\overline{Y}\right]\geq p_1-\zeta_n$. Since as $n$ becomes large, $\zeta_n\rightarrow 0$ and $p_1 \in (0,1)$, we can conclude
\begin{align}\label{eq:H4}
\lim\limits_{n\to\infty} E\left[\overline{Y}\right]=p_1.
\end{align}
Also,
\begin{align}
	\no Var\left(\overline{Y}\right) &= \frac{1}{\left(N^{(n)}\right)^2}\sum_{u=1}^{N^{(n)}}p_u\left(1-p_u\right) \\
	\nonumber &\leq \frac{1}{(N^{(n)})^2} \sum_{u=1}^{N^{(n)}}\left(p_1+\zeta_n\right)\left(1-p_1+\zeta_n\right)\\
	\nonumber &= \frac{1}{(N^{(n)})^2}\cdot N^{(n)}\left(p_1+\zeta_n\right)\left(1-p_1+\zeta_n\right) \\
	\nonumber &= \frac{1}{N^{(n)}} \left(p_1+\zeta_n\right)\left(1-p_1+\zeta_n\right). \ \
\end{align}
Thus,
\begin{align}\label{eq:H5}
\lim_{n\to\infty} Var\left(\overline{Y}\right)=0.
\end{align}
By using (\ref{eq:H4}), (\ref{eq:H5}), and Chebyshev's inequality, we can conclude
\[\overline{Y}\xrightarrow{d} p_1.\]

\end{proof}